\documentclass{article}

\usepackage{arxiv}

\usepackage[utf8]{inputenc} 
\usepackage[T1]{fontenc}    
\usepackage{hyperref}       
\usepackage{url}            
\usepackage{booktabs}       
\usepackage{amsfonts}       
\usepackage{nicefrac}       
\usepackage{microtype}      
\usepackage{lipsum}		
\usepackage{graphicx}
\usepackage{natbib}
\usepackage{doi}
\usepackage{amsmath}
\usepackage{amssymb}
\usepackage{xcolor}
\usepackage{mathtools}
\usepackage{amsthm}
\usepackage{enumitem}
\usepackage{svg}
\usepackage{algorithm2e}
\usepackage{multicol, multirow}
\usepackage{arydshln}
\usepackage{booktabs}
\usepackage{caption,subcaption}

\newcommand{\qref}[1]{Eq.~(\ref{#1})}

\title{
Unbalanced optimal transport for stochastic particle tracking
}

\newtheorem{definition}{Definition}
\newtheorem{remark}{Remark}

\newtheorem{lemma}{Lemma}
\newtheorem{theorem}{Theorem}[section]

\author{{\includegraphics[scale=0.06]{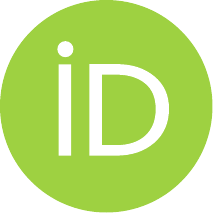}\hspace{1mm}Kairui~Hao} \\
	School of Mechanical Engineering\\
	Purdue University\\
	West Lafayette, IN \\
	\texttt{hao55@purdue.edu} 
 \\ 
	\And
 {\includegraphics[scale=0.06]{orcid.pdf}\hspace{1mm}Atharva Hans} \\ 
	School of Mechanical Engineering\\
	Purdue University\\
	West Lafayette, IN \\
	\texttt{hans1@purdue.edu} \\
	\And
{\includegraphics[scale=0.06]{orcid.pdf}\hspace{1mm}Pavlos Vlachos} \\ 
	School of Mechanical Engineering\\
	Purdue University\\
	West Lafayette, IN \\
	\texttt{pvlachos@purdue.edu} \\
	\And
 {\includegraphics[scale=0.06]{orcid.pdf}\hspace{1mm}Ilias~Bilionis}\thanks{Corresponding author.} \\ 
	School of Mechanical Engineering\\
	Purdue University\\
	West Lafayette, IN \\
	\texttt{ibilion@purdue.edu} \\
}

\hypersetup{
pdftitle={},
pdfsubject={q-bio.NC, q-bio.QM},
pdfauthor={David S.~Hippocampus, Elias D.~Striatum},
pdfkeywords={First keyword, Second keyword, More},
}

\begin{document}
\maketitle

\begin{abstract}
Non-invasive flow measurement techniques, such as particle tracking velocimetry, resolve 3D velocity fields by pairing tracer particle positions in successive time steps. 
These trajectories are crucial for evaluating physical quantities like vorticity, shear stress, pressure, and coherent structures.
Traditional approaches deterministically reconstruct particle positions and extract particle tracks using tracking algorithms.
However, reliable track estimation is challenging due to measurement noise caused by high particle density, particle image overlap, and falsely reconstructed 3D particle positions. 
To overcome this challenge, probabilistic approaches quantify the epistemic uncertainty in particle positions, typically using a Gaussian probability distribution.
However, the standard deterministic tracking algorithms relying on nearest-neighbor search do not directly extend to the probabilistic setting.
Moreover, such algorithms do not necessarily find globally consistent solutions robust to reconstruction errors.
This paper aims to develop a globally consistent nearest-neighborhood algorithm that robustly extracts stochastic particle tracks from the reconstructed Gaussian particle distributions in all frames.
Our tracking algorithm relies on the unbalanced optimal transport theory in the metric space of Gaussian measures.
Specifically, we optimize a binary transport plan for efficiently moving the Gaussian distributions of reconstructed particle positions between time frames.
We achieve this by computing the partial Wasserstein distance in the metric space of Gaussian measures.
Our tracking algorithm is robust to position reconstruction errors since it automatically detects the number of particles that should be matched through hyperparameter optimization.
Notably, our tracking algorithm also readily applies to the standard deterministic PTV case.
Finally, we validate our method using an in vitro flow experiment using a 3D-printed cerebral aneurysm.
\end{abstract}


\section{Introduction}

Particle tracking velocimetry (PTV) is a fluid velocity field measurement technique that works by tracking the tracer particles
\citep{maas1993particle,malik1993particle}.
Multiple cameras record the three-dimensional motion of tracer particles in two-dimensional images.
Then one uses a regression approach to build forward measurement models that map the physical space particles to the images \citep{maas1993particle,malik1993particle}.
The two-dimensional projected particle images can be mapped back to the physical space using triangulation \cite{maas1993particle,wieneke2012iterative}.
Once physical space particle positions are reconstructed at each recorded frame, one uses tracking algorithms to obtain Lagrangian tracks \cite{malik1993particle}.
Finally, the Eulerian velocity and pressure fields can be estimated from these Lagrangian tracks \citep{neeteson2016pressure, zhang2020using,virant19973d,hagemeier2015comparative}.

Two crucial steps in PTV are the
reconstruction of physical space particle positions from recorded images and the extraction of Lagrangian tracks from the reconstructed particle positions.
One of the standard reconstruction methods is
Iterative Particle Reconstruction \citep{wieneke2012iterative, jahn2021}. 
To extract Lagrangian tracks, one uses tracking algorithms to identify the most probable tracks.
The classic approach is the nearest neighbor searching (NNS) algorithm \cite{malik1993particle} that finds the nearest particle to an individual particle or predictor location.
Several subsequent improvements to the NNS algorithm have been developed.
\citet{dracos1996particle} proposed to penalize large acceleration.
\citet{baek1996new} used iterative estimation of match probability and no-match probability.
\citet{okamoto1995new} developed a spring model technique to match particle clusters.
The method in 
\cite{guezennec1994algorithms} finds the most likely particle tracks by minimizing a penalty function associated with each possible track.
The objective is to use path coherence, such as smoothness of position and velocity to identify particle tracks.
\cite{li2008multi} used a regression method to predict future particle positions and the developed method is more robust to noisy input particle positions.
\citet{mikheev2008enhanced} added a term that accounts for particle diameters for enhancing pairwise matching.
\cite{cardwell2011multi} introduced a multi-parametric particle-pairing algorithm.
\cite{guo2019microscale} extended this method to a generalized multi-parametric pairing algorithm that uses more features, such as diameter, peak intensity, size, aspect ratio, or fluorescence color.
\cite{cierpka2013higher} uses a multi-frame high-order approach to pair particles.
Instead of finishing reconstruction and tracking in two steps, the state-of-the-art method, shake-the-box \cite{schanz2016shake}, combines these two steps together to identify the Lagrangian tracks in time progressively.

Despite the great progress in the field, there still remain some open challenges.
First, all current methods are deterministic and cannot account for uncertainties in PTV.
In the reconstruction step, the reconstructed physical space particles are not fully faithful due to camera measurement noise, overlapping particles, and reconstruction algorithm errors.
These errors subsequently propagate to the particle tracking phase.
Deterministic approaches produce a single Lagrangian track, which is incapable of quantifying the errors.
Second, all the tracking algorithms use NNS-type approaches that only utilize local displacement information to match particle pairs. 
This local matching strategy is not necessarily globally consistent since
one particle might be identified as the nearest neighbor particle for multiple particles or predictors.
Hence, further processing is required to identify the unique particle tracks.
Last, it is still challenging for the current tracking algorithms to address the reconstruction errors,
and robustly identify Lagrangian tracks from the potentially erroneous position reconstruction results.

Uncertainty quantification has been studied extensively for particle image velocimetry~\citep{timmins2012method,charonko2013estimation,xue2014particle,sciacchitano2015collaborative,xue2015particle,kahler2016main,
bhattacharya2016stereo,boomsma2016comparative,bhattacharya2018particle,rajendran2021meta}.
Recent work on PTV particle reconstruction methods has also shown a growing interest in uncertainty quantification.
The method can be categorized into non-Bayesian and Bayesian.
In a non-Bayesian approach,
\citet{bhattacharya2020volumetric} quantify the uncertainty in particle image positions and the mapping function coefficients.
The approach relies on the frequentist perspective to find the covariance matrices from the training data.
Subsequently, these uncertainties can be propagated back using first-order Taylor expansion to get the reconstructed particle positions and velocity field.
Our research group's previous work \cite{hansstochastic} employed the Bayesian Volumetric Reconstruction framework to get the reconstructed particle positions.
The methodology first assumes a prior distribution of the position of each particle.
Such a prior distribution quantifies the belief of particle positions before observing images.
In the simplest case, one can choose a uniform distribution as the prior position distribution for each particle.
Next, it uses the calibrated camera measurement mapping function and recorded images to derive a likelihood function.
Finally, the authors use the variational inference \citep{jordan1999introduction,blei2017variational} to update the posterior distribution of the reconstructed particle positions.
These stochastic reconstruction methods estimate the distributions of particle positions using Gaussian measures.
However, to the best of our knowledge, there does not exist a particle tracking algorithm that can extract stochastic Lagrangian tracks from stochastically reconstructed particle positions.
Robustly identifying the stochastic Lagrangian tracks against reconstruction errors is crucial to reconstruct the stochastic Eulerian velocity field using Physics-informed reconstruction algorithms \citep{alberts2023physics, hao2024information}.

Optimal transport (OT) theory \citep{kantorovich2006translocation,  villani2009optimal,villani2021topics} studies the optimal transportation and allocation of resources.
It is a rather mathematical subject \citep{santambrogio2015optimal}, but has been successfully applied to practical problems such as full-waveform inversion \citep{yang2018application}, economics \citep{galichon2018optimal}, data science \citep{peyre2019computational}, and deep learning \citep{arjovsky2017wasserstein}.
The application of OT to PTV is limited.
\citep{agueh2015optimal} and \citep{saumier2015optimal} first brought the OT method to particle image velocimetry.
They studied 2D velocity field reconstruction from images using OT.
However, their OT approach is incompatible with the recent computational optimal transport theory and algorithms, which have progressed significantly.
\citep{chen2023unbalanced} formulated a regularized unbalanced optimal transport (UOT) problem to analyze fluid flows in the brain.
The UOT problem is a continuous variational problem subject to a partial differential equation constraint.
This has a great theoretical interest but is numerically expensive.
A similar work to ours is 
\citep{liang2023recurrent}.
The author introduced a complete pipeline to learn the 3D fluid motion from reconstructed particle positions.
They used neural networks to transform particle sets into feature spaces and solve the entropy-regularized UOT problem in the feature spaces. 
Subsequently, they adaptively retrieved the optimal transport plan for iterative flow updating.
We argue that the most noticeable advantage of OT over NNS is that OT finds particle pairs by combining local and global displacement information.
However, all of the aforementioned OT-based PTV approaches are not able to quantify and propagate uncertainties.
Moreover, it has been shown that OT is not robust to outliers since it uses all possible particle displacements to produce the optimal transport plan \citep{balaji2020robust, mukherjee2021outlier,le2021robust,nietert2022outlier}.
The falsely reconstructed particle positions will trigger OT to replan the transport that accounts for these undesirable outlier particles. 
It is not clear how the aforementioned applications of OT to PTV address the robustness issue when facing reconstruction errors.

The objective of this paper is to develop a globally consistent nearest-neighborhood algorithm that robustly extracts stochastic particle tracks from the reconstructed Gaussian particle distributions in all frames. 
First, we formulate an integer UOT (IUOT) problem, which is a binary linear programming (BLP) problem 
that achieves one-to-one pairing.
Next, we solve the partial Wasserstein distance (PWD) \citep{chapel2020partial} formulation of the IUOT.
We show that PWD is a relaxed version of the BLP that can be solved in polynomial time.
We give the theoretical guarantee that the relaxed BLP attains the same optimal point of BLP.
Our method is the first PTV tracking algorithm to match reconstructed Gaussian position estimation and propagate uncertainties from stochastic particle positions to stochastic Lagrangian tracks.
We achieve this by matching the Gaussian position estimation in the metric space of Gaussian measures, also known as the Wasserstein space.
Our method is the first PTV tracking algorithm that explicitly addresses the tracking robustness when facing falsely reconstructed particle positions.
The algorithm achieves this by automatically detecting the unknown number of particle pairs that should be matched through hyperparameter optimization.

The structure of this paper is as follows.
In section \ref{section:background} we review the Bayesian volumetric reconstruction method and OT/UOT theory.
In section \ref{section:methodology}, we develop our tracking algorithm.
In section \ref{section:examples},
we verify and validate the algorithm with a synthetic 3D Burgers vortex and a cerebral aneurysm flow experimental examples.
In section \ref{section:conclusions}, we conclude the paper.

\section{Background}
\label{section:background}

\subsection{Bayesian volumetric reconstruction in PTV}

We review the Bayesian volumetric reconstruction (BVR) method developed in~\citep{hansstochastic}.
Define $U$ as a 3D flow domain and $r=(x, y, z)\in U$ as a position in the domain. 
BVR formulates a Bayesian inference problem to estimate the posterior distribution of the 3D physical space particle positions given observed camera images.
Two ingredients are needed.
First, a prior distribution of 3D particle positions $p(r_{1:N})$ that quantifies our belief of particle positions before observing the images.
This is the users' choice and can be just a uniform distribution over $U$.
Second, a likelihood function that links the 3D particle positions to the observed images.
Two basic setups are required in the likelihood function:  forward and  measurement models.

\subsubsection{Forward camera model}
To develop a forward model, all cameras are first calibrated to map the 3D physical space particle positions $r$ to 2D image coordinates.
The mapping functions are typically two polynomial functions \citep{soloff1997distortion, wieneke2005stereo} per camera, which map $r$ to horizontal and vertical coordinates in the image, respectively.
The typical choice is two 3rd-order polynomials.
In the multi-index notation, we can write, for the $n$th particle and the $m$th camera, the image horizontal $f_{m1}(r_n)$ and vertical $f_{m2}(r_n)$ coordinate mapping functions
\begin{align*}
    f_{mi}(r_n) = \sum_{\left\vert
    \alpha
    \right
    \vert
    \le3
    }
    w_{mi\alpha}r_n^{\alpha}, 
    \quad
    i=1\ \text{and}\ 2.
\end{align*}
In this calibration step, the polynomial coefficients $w_{mi\alpha}$ are optimized to minimize the distance between the computed image coordinates and the observed image coordinates.

Next, an intensity model describes the projected particle intensity on the image. 
This model assumes the intensity follows a Gaussian kernel function centered in the projected particle coordinate $f_{m}(r_n)=\left(f_{m1}(r_n), f_{m2}(r_n)\right)$.
So, the particle intensity peaks in its projected center and attenuates away.
Three required parameters are center intensity, stretch length, and rotation angle.
For the $n$th particle and $m$th camera, the center intensity of the projected particle image considers the effects of particle radius $\rho_n$, peak intensity $\phi_m$, slope parameter $\alpha_m$ and threshold parameter $\beta_m$.
This can be modeled using the Sigmoid function 
\begin{align*}
    \hat{I}
    \left(
        \rho_n;
        \phi_{m}, \alpha_m, \beta_{m}
    \right)
    =
    a+b\operatorname{Sigmoid}\left(
    \alpha_m(\rho_n-\beta_m)
    \right),
\end{align*}
where $a$ and $b$ are selected such that $
\lim_{\rho_n \to 0}\hat{I}
    \left(
        \rho_n;
        \phi_{m}, \alpha_m, \beta_{m}
    \right) = 0
$ and 
$
\lim_{\rho_n \to \infty} \hat{I}
    \left(
        \rho_n;
        \phi_{m}, \alpha_m, \beta_{m}
    \right)= \phi_m
$.
The stretch and rotation effects can be modeled using the diagonal scaling and rotation matrices:
\begin{align*}
    S_m
    &=\operatorname{diag}\left(
    \gamma_{m1}, \gamma_{m2}
    \right),
    \\
    Q_m&=
    \begin{bmatrix}
        \cos\left(\theta_m \right)
        &
        -\sin\left(\theta_m \right)
        \\
        \sin\left(\theta_m \right)
        &
        \cos\left(\theta_m \right)
    \end{bmatrix}.
\end{align*}
The projected $n$th particle intensity on $(i, j)$th pixel $p_{mij}$ of the $m$th camera  is
\begin{align*}
I_{mij}\left(r_n\right)
=
    \hat{I}
    \left(
        \rho_n;
        \phi_{m}, \alpha_m, \beta_{m}
    \right)
    \exp\left\{
        -\frac{1}{2}
        (p_{mij}-f_m(r_n))^T
        Q_m S_m Q_m^{T}
        (p_{mij}-f_m(r_n))
    \right\}.
\end{align*}
The total projected intensity on $(i, j)$th pixel of the $m$th camera  is the sum of $N$ individual projected particle intensities at this pixel:
\begin{align*}
    I_{mij}\left(r_{1:N}\right) =
    \sum_{n=1}^N
     I_{mij}(r_n).
\end{align*}

\subsubsection{Measurement function}

The measurement function describes how the observed $m$th grayscale camera image $y_m$ is generated by the projected intensity $I_{mij}\left(r_{1:N}\right)$.
Without loss of generality, we normalize $y_m$ between 0 and 1.
Then we make the conditionally independent assumption, and the likelihood of the $m$th camera follows
\begin{align*}
    p\left(y_m\vert r_{1:N}\right)
    =
    \prod_{i}^{I}\prod_{j}^J
    p\left(
    y_{mij}|r_{1:N}
    \right),
\end{align*}
where each pixel grayscale color follows the truncated Normal distribution with the standard deviation $\delta_m$
\begin{align*}
    p\left(
    y_{mij}|r_{1:N}
    \right)
    \sim 
    \operatorname{TruncatedNormal}_{(0, 1)}\left(
        I_{mij}\left(r_{1:N}
        \right),
        \delta_{m}^2
    \right).
\end{align*}
To emphasize the parameters defined in the forward camera model $I_{mij}\left(r_{1:N}\right)$
and the measurement standard deviations, 
we collectively write them as $
\lambda=
\left\{
\rho,
\phi,
\alpha,
\beta,
\gamma,
\theta, 
\delta
\right\}
$.
Then the likelihood for all $M$ number of cameras is the product of the likelihoods of each camera:
\begin{align*}
    p_{\lambda}\left(y_{1:M}\vert r_{1:N}\right)
    =
    \prod_{m=1}^M
    p_{\lambda}\left(y_m\vert r_{1:N}\right).
\end{align*}

\subsubsection{Variational inference}

The posterior distribution of 3D particle positions quantifies the epistemic uncertainty after observing the camera images.
Using Bayes' rule,  it is
\begin{align*}
    p_{\lambda}(r_{1:N}\vert y_{1:M})
    \propto
    p_{\lambda}\left(y_{1:M}\vert r_{1:N}\right)
    p(r_{1:N}).
\end{align*}
Jointly finding the posterior distribution and the optimal likelihood parameter $\lambda$ analytically is intractable.
So BVR uses variational inference \citep{jordan1999introduction} and the training technique in variational auto-encoder \citep{kingma2013auto} to approximate the posterior distribution while finding the optimal parameter. 

This is achieved by maximizing the evidence lower bound (ELBO) over the likelihood parameter $\lambda$ and a family of parameterized distribution $q_{\psi}(r_{1:N})$ to approximate the posterior $ p_r(r_{1:N}\vert y_{1:M})$.
Mathematically, the ELBO is defined as
\begin{align*}
    \operatorname{ELBO}\left(\psi, \lambda\vert y_{1:M}\right)
    =
    \int 
    \log\left\{
    \frac{
     p_{\lambda}\left(y_{1:M}\vert r_{1:N}\right)
    p(r_{1:N})
    }{
    q_{\psi}(r_{1:N})
    }
    \right\}
    q_{\psi}(r_{1:N})\ dr_{1:N}.
\end{align*}

To improve the computational efficiency, BVR chooses the diagonal multivariate normal distribution for each particle position as the guide:
\begin{align*}
    q_{\phi}(r_{1:N})
    =
    \prod_{n=1}^N
    \mathcal{N}
    \left(
    r_n\vert
    \mu_n, \operatorname{diag}
    \left(
    \sigma^2_{n}
    \right)
    \right),
\end{align*}
where $\mu_n$ and $\sigma_n$ are the three-dimensional mean and standard deviation vectors for the $n$th particle, respectively.
To further improve the reconstruction accuracy, BVR adds a penalty term to the ELBO.
This significantly helps to escape local minimums in the optimization.
The parameters $\lambda$ defined in the forward camera model and measurement function are also optimized while computing the posterior distribution.
The reader can find the details in~\citep{hansstochastic}.

\subsection{Balanced and unbalanced optimal transport}

The optimal transport problem originates from the question of how to best move piles of sand to fill up holes of the same total volume.
Monge \citep{villani2009optimal, peyre2019computational}
first formulated the problem as finding the optimal transport plan.
Formally, let $X$ and $Y$ be two metric spaces, $\mu$ and $\nu$ be two probability measures on $X$ and $Y$, respectively, and $c: X \times Y \to [0, \infty]$ be a cost function of moving a mass from a point $x$ to $y$. Monge's formulation of the optimal transportation problem is to find a pushforward transport map $T: X\to Y$ that realizes the infimum:
\begin{align*}
    \inf_{T} \left\{
        \int_{X} c\left(x, T(x)\right)d \mu(x)
        \middle\vert
        T_{\#}(\mu)=\nu
    \right\}.
\end{align*}
The constraint $T_{\#}(\mu)=\nu$ means $\nu$ is the pushforward measure of $\mu$ under the transport map $T$ \citep{stein2009real}.
This means, for any Borel set $A$ in $Y$, we have
$
\nu(A) = \mu(T^{-1}(A))
$.
However, Monge's formulation exhibits several difficulties. The objective functional and constraint are not convex, and the transport map might not exist \citep{santambrogio2015optimal}.

Kantorovich's formulation of the optimal transportation problem 
relaxes Monge's formulation by allowing mass split and has a probabilistic interpretation.
Formally, let $\Gamma(\mu, \nu)$ denote the collection of all probability measures on $X\times Y$ with marginals $\mu$ on $X$ and $\nu$ on $Y$. 
We define the projection operators:
\begin{align}
    P^{X}_{\#} r
    =
    \int_{Y} \gamma(x, y) dy
    \quad
    \text{and}
    \quad
    P^{Y}_{\#} r
    =
    \int_{X} \gamma(x, y) dx. \label{projections}
\end{align}
Then $\mu = P^{X}_{\#} r$ and $\nu = P^{Y}_{\#} r$.

Kantorovich's formulation aims to find a coupling measure $\gamma$ that realizes the infimum:
\begin{align}
    \inf_{\gamma} \left\{
    \int_{X\times Y} c(x, y) d\gamma(x, y)
    \middle\vert 
    \gamma\in \Gamma(\mu, \nu)
    \right\},
    \label{kantor}
\end{align}
This coupling measure has the interpretation of a transport plan~\cite{arjovsky2017wasserstein}. 

A special but useful case is when both $\mu$ and $\nu$ are discrete unnormalized uniform measures:
\begin{align*}
    \mu(x) = \sum_{i=1}^N \delta_{x_i}(x), \ \text{and}\
    \nu(y) = \sum_{i=1}^N \delta_{y_i}(y).
\end{align*}
In this case, the cost function $c(x,y)$ and the coupling measure $\gamma(x,y)$ have finite-dimensional matrix representations.
Specifically, let $c_{ij}$ and  $\gamma_{ij}$
denote the cost of moving a unit mass and the total mass to be transported from the $i$th source to the $j$th target, respectively.
The objective function can be formulated using the Frobenius product 
$\langle \gamma, c\rangle_F \coloneqq \sum_{i, j} \gamma_{ij}c_{ij}$.
Then one can translate \qref{kantor} into the finite-dimensional linear programming problem:
\begin{equation}
    \begin{split}
        \min_{\gamma}\quad &
        W=
        \langle \gamma, c\rangle_F,
     \\
     s.t. \quad &\sum_{j=1}^N \gamma_{ij} = 1, \ \text{for}\ i=1, \cdots, N,
     \\
      &\sum_{i=1}^N \gamma_{ij} = 1, \ \text{for}\ j=1, \cdots, N,
      \\
     & \gamma_{ij}\ge 0.
    \end{split}
    \label{balanced_ot}
\end{equation}
The first two constraints are a finite-dimensional version of the marginal consistent constraints.
Solving this optimization problem produces an optimal transport plan matrix $\gamma^{\ast}$.
Due to the Birkhoff-von Neumann \citep{jurkat1967term} and the fundamental theorem of linear programming \citep{bertsekas2009convex}, the optimal $\gamma^{\ast}$ only contains 0 and 1.
Computing the optimal transport cost using classic algorithms such as network simplex or interior point scales at least in $O(N^3\log(N))$ \cite{pele2009fast}.
So it becomes quickly computationally prohibitive.
\citep{cuturi2013sinkhorn} revitalized the computational optimal transport by defining an entropy-regularized optimal transport problem.
This problem can be solved by the Sinkhorn algorithm, which scales in $O(N^2)$ \cite{sinkhorn1967diagonal, knight2008sinkhorn, lin2022efficiency}.

One can use \qref{kantor} to define the Wasserstein $p$-distance (with respect to the distance $d(x, y)$)  between two probability measures $\mu$ and $\nu$ with finite $p$-moments using the expectation operator $\mathbb{E}$:
\begin{align*}
    W_p\left(\mu, \nu\right)
    =
    \left(
    \inf_{\gamma\in \Gamma(\mu, \nu)}
    \mathbb{E}_{(x, y)\sim \gamma}
    d(x, y)^p
    \right)^{1/p}.
\end{align*}
Two special cases we will use later are Dirac measures and Gaussian measures with the Euclidean norm $d(x, y)= \left\Vert x-y\right\Vert$.
Specifically, for the Dirac measures, we have
\begin{align*}
     W_2\left(
      \delta_{x},
    \delta_{y}
     \right)^2
     =
     \left\Vert x-y\right\Vert^2,
\end{align*}
and for the Gaussian measures \citep{olkin1982distance}
 \begin{align}
     W_2\left(
      \mathcal{N}(m_1, \Sigma_1),
    \mathcal{N}(m_2, \Sigma_2)
     \right)^2
     =
     \left\Vert
     m_1-m_2
     \right\Vert^2_2
     +
     \operatorname{trace}
     \left(
     \Sigma_1+ \Sigma_2 -
     2(
        \Sigma_2^{1/2}\Sigma_1
        \Sigma_2^{1/2}
     )^{1/2}
     \right), \label{normal_was_full}
 \end{align}
where $\Sigma^{1/2}_2$ is the principal square root \citep{horn2012matrix} of $\Sigma_2$.
When the covariance matrix is diagonal, the squared Wasserstein-2 distance can be simplified to:
 \begin{align*}
     W_2\left(
      \mathcal{N}(m_1, \operatorname{diag}(\sigma_1^2)),
    \mathcal{N}(m_2, \operatorname{diag}(\sigma_2^2))
     \right)^2
     =
     \left\Vert
     m_1-m_2
     \right\Vert^2_2
     +
     \left\Vert
     \sigma_1-\sigma_2
     \right\Vert^2_2.
 \end{align*}

 The variational problem defined in \qref{kantor} is a balanced optimal transport since the total masses of the source and target are equal.
 More often than not, the source and target masses are not equal in practical problems \cite{fatras2021unbalanced, balaji2020robust}.
 The theoretical and numerical development of UOT has advanced since the two papers \cite{chizat2018unbalanced, chizat2018scaling}.
 In the UOT, the hard marginal constraints \qref{projections} are relaxed by using divergence between probability measures.
 The variational formulation of UOT has the form:
 \begin{align}
     \inf_{\gamma}
     \left\{
     \int_{X\times Y}
     c(x, y) d \gamma(x, y)
     +
     D_{X}\left(
     P^{X}_{\#}\gamma \middle \vert \mu
     \right)
     +
     D_{Y}\left(
     P^{Y}_{\#}\gamma \middle \vert \nu
     \right)
     \right\},\label{uot_original}
 \end{align}
 where $D_X$ and $D_Y$ are convex divergences such as the Kullback-Leibler divergence \citep{kullback1951information} that measure the discrepancies of the marginals of $\gamma$ between $\mu$ and $\nu$, respectively.
 The divergences play as a soft penalty to match the marginal consistent constraints.
 One can also vary the weights of the two divergences to determine how much the marginal constraints should be satisfied.
\section{Methodology}

\label{section:methodology}

In this section, we use UOT to develop a  stochastic particle 
tracking algorithm that can identify and extract stochastic particle tracks from the posterior particle distributions produced by BVR.
We first present the matching algorithm that pairs Gaussian distributions between two frames.
To improve the matching performance, we will discuss the first-order position prediction and hyperparameter optimization.
Next, we discuss how to extract stochastic Lagrangian tracks after Gaussian distribution pairs having been identified.

\subsection{Robust particle matching using integer unbalanced optimal transport}

We first explain how to robustly match noisy reconstructed particles between two consecutive time frames and propagate uncertainties.
In the following sections, we assume reconstructed particle positions at each time frame are provided through particle position reconstruction.
Particle position reconstruction approaches generally fall into two categories: deterministic \citep{wieneke2012iterative, schanz2016shake} and stochastic approaches \citep{bhattacharya2020volumetric, hansstochastic}.
In the deterministic approaches, reconstructed particle positions 
can be represented by the sum of Dirac measures
$
\sum_{i=1}^{N} \delta_{r_i}(r)
$.
In the stochastic approaches, this becomes the sum of Gaussian measures
$
\sum_{i=1}^{N} 
\mathcal{N}\left(r\middle\vert m_i, \Sigma_i\right).
$
It is important to recall that  Gaussian measure is fully supported in $\mathbb{R}^3$, while the flow domain $U$ is usually supported finitely.
In practical applications, it is still reasonable to represent a particle position using a Gaussian measure since particles that are outside of the boundary of $U$ are discarded.

Although the objective of this paper is to develop a matching algorithm that works for stochastically reconstructed particles, our algorithm also suitably applies to the deterministic reconstruction results. 
To cover these two cases, we work in the metric space of probability measures $(\mathbf{P}, d)$.
A point in this metric space is denoted by $\pi$ and a collection of probability measures from this metric space is denoted by 
$\left\{
\pi_i
\right\}_{i\in I}$.
In the deterministic reconstruction case,
$\mathbf{P}$ is the space of Dirac measures.
A point $\pi$ is a Dirac measure, and the metric $d$ is the usual Euclidean norm.
In the stochastic reconstruction case,
$\mathbf{P}$ is the space of Gaussian measures, a point $\pi$ is a Gaussian measure and the metric $d$ is the Wasserstein 2- distance \citep{takatsu2011wasserstein}.
We assume there are $N$  reconstructed particle positions
$
\Pi^k=
\left(
\pi^k_1, \cdots, \pi^k_{N}
\right)
$ at time frame $k$,
and $M$ such particles
$
\Pi^{k+1}=
\left(
\pi^{k+1}_1, \cdots, \pi^{k+1}_{M}
\right)
$ at time frame $k+1$.
It should be emphasized that $N$ is not necessarily equal to $M$.

Our objective is to design an optimal transport plan in the space of probability measures $\mathbf{P}$, i.e., $X=Y=\mathbf{P}$.
First, we write down the source mass:
\begin{align}
    \mu(\pi) =
    \sum_{i}^N \delta_{\pi^k_i}\left(\pi\right),
    \label{source}
\end{align}
and the target mass:
\begin{align}
    \nu(\pi) =
    \sum_{i}^M \delta_{\pi^{k+1}_i}\left(\pi\right).\label{target}
\end{align}
A graphic illustration of transporting Gaussian distributions from the source to the target is shown in Fig. \ref{Gaussian_transport}.

\begin{figure}[!h]
  \centering
\includegraphics[scale=0.45]{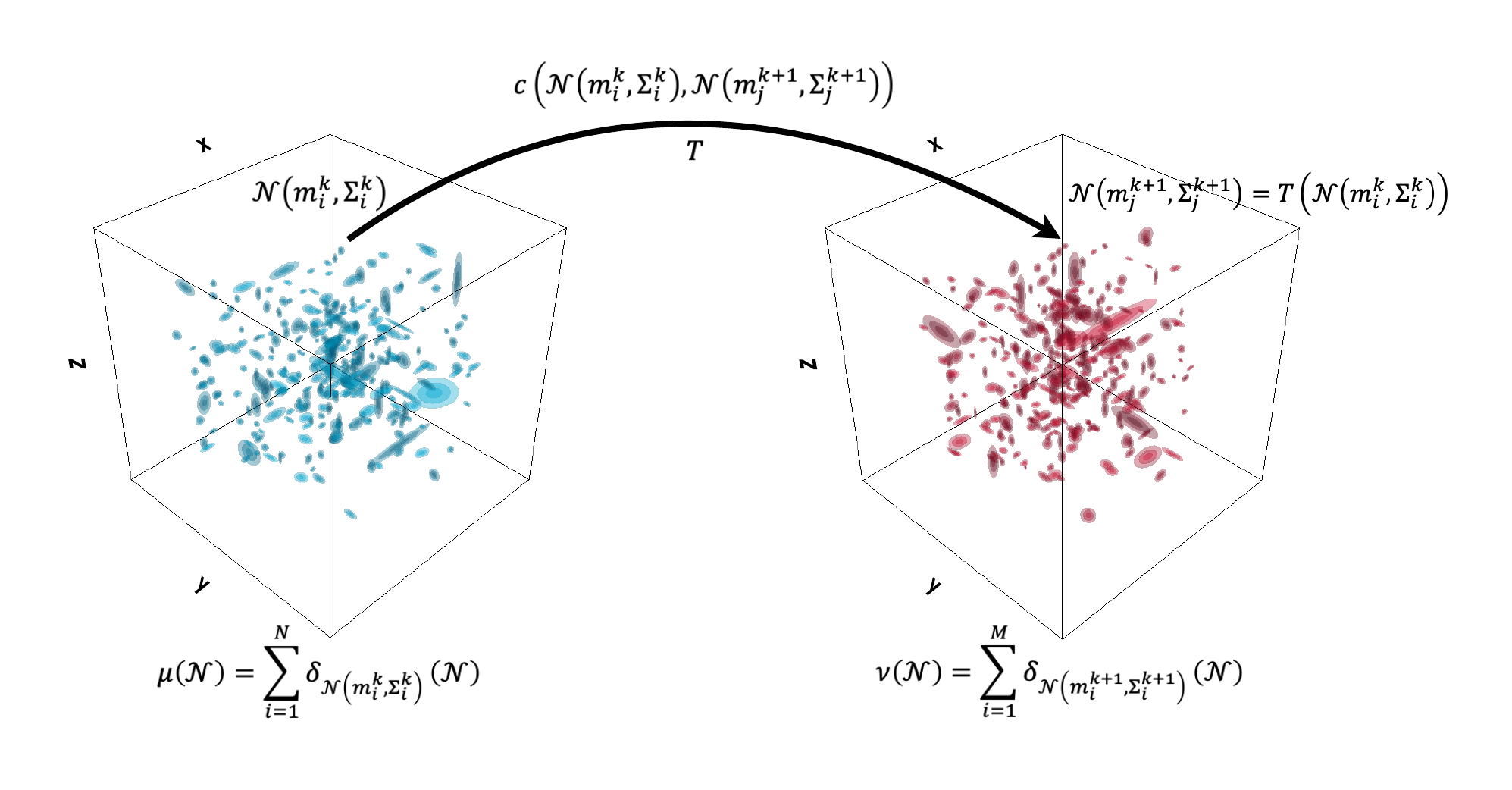}
  \caption{An illustration of unbalanced optimal transport of stochastic particle position reconstructions.}
  \label{Gaussian_transport}
\end{figure}

Since the source mass is not necessarily equal to the target mass, i.e., $N\neq M$, this problem is unbalanced. 
The original UOT formulation in \qref{uot_original} and extensions
\citep{chapel2021unbalanced} allow mass split.
Therefore, the optimal transport matrix is a dense matrix that is hard to process.
In addition, the total transported mass is unknown a priori.
Inspired by the partial Wasserstein distance (PWD) formulation of UOT \citep{chapel2020partial}, we formulate an IUOT that is a BLP problem
shown in Eqs.~(\ref{buot}), where we enforce that the transport plan matrix $\gamma_{ij}$ can only take 0 or 1.
Our approach computes the exact solution of UOT without adding the entropy regularization. 
In addition, our formation also constrains the total number of transported particles $N_p$.
The first constraint states that the mass transported from the source should not be greater than the mass of the source.
 The second constraint states that the mass received by the target should not be greater than the mass of the target.
 The fourth binary constraint guarantees no mass split.
 In the third constraint, 
 the positive integer $N_p$ is the total mass to be transported. 
 This is the hyperparameter 
 that should be carefully chosen or optimized.
 In PTV particle matching, due to missing and overlapped particles, fluid flowing in and out of the observation volume and particle position reconstruction errors,
 the value of $N_p$ should be chosen less than $\min\left\{N, M\right\}$.
 So we can use the relative transport number  $\alpha$ to determine $N_{p, \alpha}= \operatorname{ceil}\left(\alpha \min\{N, M\}\right)$.
 The number $N-N_{p,\alpha}$ denotes our guess of the number of particles that do not have a matched particle pair in frame $k+1$ due to position reconstruction errors or particles flowing out of the flow domain.
 Similarly, the number $M-N_{p,\alpha}$ denotes our guess of the number of particles that do not have a matched particle pair in frame $k$.
 Therefore the hyperparameter $\alpha$ determines the robustness of our algorithm.
 Intuitively, a smaller $\alpha$ is more robust but identifies less particle tracks (details in section \ref{section_interpret}) and vice versa.
 In section \ref{section_hyper}, we will discuss the hyperparameter optimization of $\alpha$ in detail.

The BLP problem is NP-hard \citep{sherali2000evolution}, and the practical approach is to relax it to a linear programming problem that can be solved in polynomial time.
The PWD formulation \citep{chapel2020partial}
shown in Eqs.~(\ref{uot}) is a relaxed version of our BLP problem, 
since it only requires that $\gamma_{ij}$ is nonnegative.
The authors built an algorithm based on Frank-Wolfe optimization scheme \citep{frank1956algorithm}
and mentioned that the algorithm produces a sparse transport plan matrix.
In our formulation, we force $N_p$ to be an integer and give an important proposition stating that the PWD and BLP have the same optimal solution.
Namely, the optimal transport matrix of PWD only contains 0 and 1.
This allows us to solve the polynomial time PWD problem to compute the optimal binary transport plan matrix. 
This will ensure that the whole mass of a probability measure is transported to one target at the next time frame, simplifying downstream processing for particle tracking.

\begin{multicols}{2}
\noindent
  \begin{equation}
    \begin{split}
    &\text{IUOT (BLP):}
    \\
    \min_{\gamma}\quad &\langle \gamma, c\rangle_F
     \\
     s.t. \quad &\sum_{j=1}^M\gamma_{ij}\leq 1, \ \text{for} \ i = 1, \cdots, N,
     \\
     &\sum_{i=1}^N\gamma_{ij}\leq 1, \ \text{for} \ j = 1, \cdots, M,
     \\
     &\sum_{i=1}^N\sum_{j=1}^M \gamma_{ij} = N_{p}\in \mathbb{N}^{+} \le \min\left\{N, M\right\},
     \\
     &\gamma_{ij}\in \{0, 1\}.
    \end{split}
    \label{buot}
\end{equation}
  \begin{equation}
    \begin{split}
    &\text{PWD (Relaxed BLP):}
    \\
    \min_{\gamma}\quad &
    PW_{N_p}=\langle \gamma, c\rangle_F
     \\
     s.t. \quad &\sum_{j=1}^M\gamma_{ij}\leq 1, \ \text{for} \ i = 1, \cdots, N,
     \\
     &\sum_{i=1}^N\gamma_{ij}\leq 1, \ \text{for} \ j = 1, \cdots, M,
     \\
     &\sum_{i=1}^N\sum_{j=1}^M \gamma_{ij} = N_{p}\in \mathbb{N}^{+} \le \min\left\{N, M\right\},
     \\
     &\gamma_{ij} \ge 0, 
    \end{split}
    \label{uot}
\end{equation}
\end{multicols}

\begin{theorem}     
\label{theorem:1}     Let $\gamma^{\ast}_{\text{IUOT}}$ and $\gamma^{\ast}_{\text{PWD}}$ be the minimizers of IUOT (Eqs.~(\ref{buot})) and PWD (Eqs.~(\ref{uot})), then      $\gamma^{\ast}_{\text{IUOT}}=\gamma^{\ast}_{\text{PWD}}$. 
\end{theorem}
Proof is in Appendix~\ref{appendix:proof}.

\subsubsection{Construction of the transport cost matrix}

In this section, we explain two approaches to constructing the transport cost matrix using zero or first-order particle position predictions.
The first-order prediction approach is an improvement of the zero-order prediction approach and internally incorporates the zero-order prediction.

\paragraph{Zero-order prediction:}
In the zero-order prediction case, 
we define the cost matrix $c_{ij} = W_2\left(\pi^k_{i}, \pi^{k+1}_j\right)^2$.
It is called zero-order prediction since we calculate the distance between $\pi^k_i$ and $\pi^{k+1}_j$ without making position prediction using $\pi^k_{i}$.
 When reconstructed particle positions are deterministic,
the cost matrix element $c_{ij} = W_2\left(
      \delta_{r_i^k},
    \delta_{r_j^{k+1}}
     \right)^2$ is the squared Euclidean norm. 
If Gaussian estimation of particle positions is provided, the cost matrix element
$
c_{ij} =W_2\left(
      \mathcal{N}(m^k_i, \Sigma^k_i),
    \mathcal{N}(m^{k+1}_j, \Sigma^{k+1}_j)
     \right)^2
$ is the squared Wasserstein-2 distance. 

One should notice that, with these transport costs, we implicitly assume that the index of the $N$ reconstructed particle position  $\pi^k$ is aligned with the row index of $\gamma$, and the index of the $M$ reconstructed particle position $\pi^{k+1}$ is aligned with the column index of $\gamma$.

The zero-order prediction performs poorly when the particle spacing is relatively small compared to the particle displacement between two frames.
This can be improved by using the first-order prediction.

\paragraph{First-order prediction:}
In the first-order prediction method, we make particle position predictions by linearly extending the matched reconstructed particle positions at frames $k-1$ and $k$. 
We use $A= \left\{1,\cdots, N \right\}$ to denote the set of particle IDs at frame $k$.
Next, we make the decomposition $A=A_{\text{matched}}\cup A_{\text{unmatched}}$.
The subsets $A_{\text{matched}}$ and $A_{\text{unmatched}}$ contain particle IDs for the reconstructed particles in $\mu$ that have and do not have matched pairs at the previous frame $k-1$, respectively. 
We use the notation $l(i)$ to denote the $l(i)$th particle at frame $k-1$ is a matched pair to the $i$th particle at frame $k$.
Namely,  $\left(\pi^{k-1}_{l(i)}, \pi^k_i\right)$ is a matched pair.

The linear extension position prediction is trivial for the particles with IDs in $A_{\text{matched}}$ since they have matched pairs at frame $k-1$.
However, for the particles with IDs in $A_{\text{unmatched}}$, we have to use neighborhood particles to estimate the predicted displacements.

First, we make the following definitions and a pictorial explanation is in Fig. \ref{epsilon particles}.
\begin{definition}[$\epsilon$-ball $B_{\epsilon}(\pi)$]
    The $\epsilon$ ball in the metric space of probability measures $(\mathbf{P}, W_2)$ centered at $\pi$ and of radius $\epsilon$ is 
    $
    B_{\epsilon}(\pi) = 
    \left\{\Tilde{\pi} \in \mathbf{P} \ \middle \vert \
    W_2\left(
    \Tilde{\pi}, \pi
    \right)
    < \epsilon
    \right\}
    $.
\end{definition}
\begin{definition}[Punctured $\epsilon$-ball $B^{\ast}_{\epsilon}(\pi)$]
    The punctured $\epsilon$-ball is the $\epsilon$-ball minus its center $
    B^{\ast}_{\epsilon}(\pi)=B_{\epsilon}(\pi)\setminus \pi
    $.
\end{definition}
\begin{definition}[$B^{\ast}_{\epsilon}(\pi)$-particles relative to a finite set of particles $\Pi$]
    The $B^{\ast}_{\epsilon}(\pi)$-particles for $\pi\in \mathbf{P}$ relative to $\Pi$ is the subset of particles from $\Pi$ contained in  $B^{\ast}_{\epsilon}(\pi)$, i.e., 
    $B^{\ast}_{\epsilon}(\pi)$-particles = 
    $
    B^{\ast}_{\epsilon}(\pi)\cap \Pi
    $.
\end{definition}

\begin{figure}[!h]
  \centering
\includegraphics[scale=0.45]{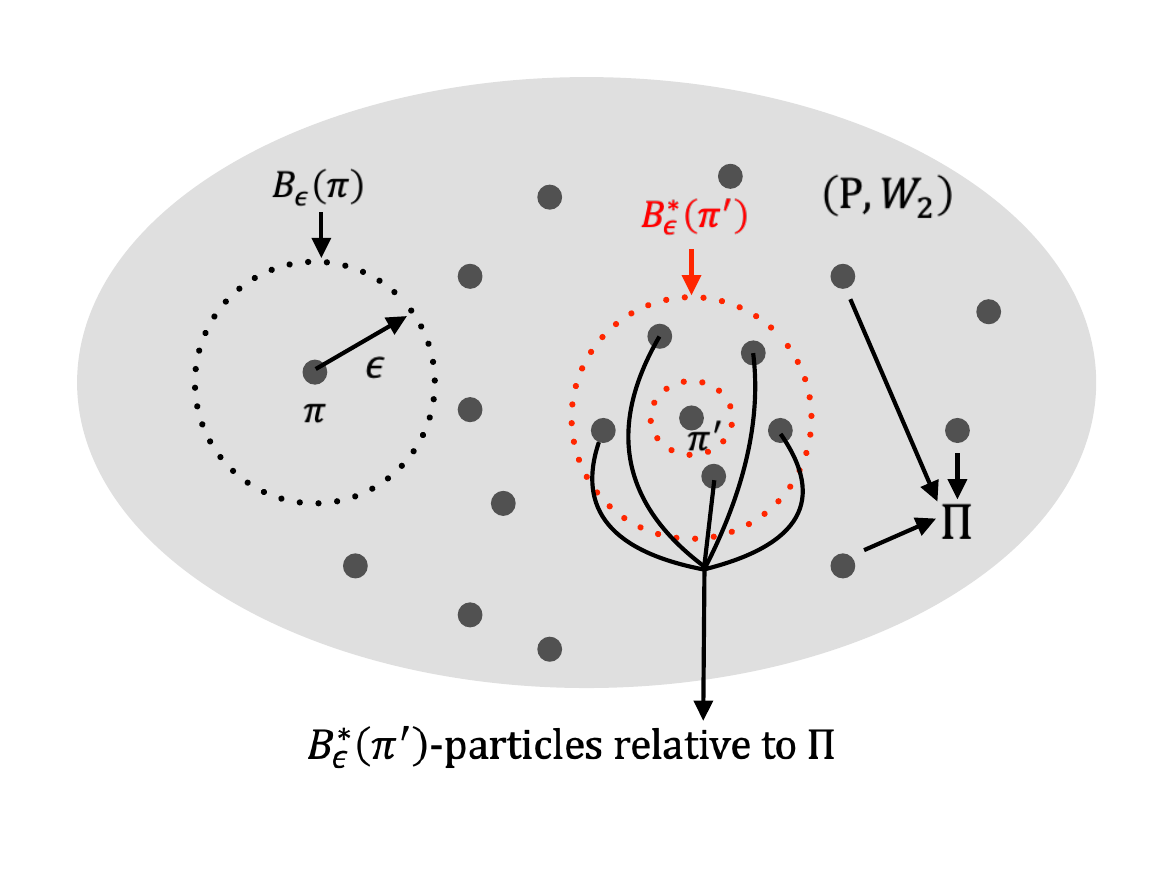}
  \caption{Visualization of particle topology in the metric space of probability measures.}
  \label{epsilon particles}
\end{figure}

The goal is to use the $B^{\ast}_{\epsilon}(\pi^k_i)$-particles (relative to the source $\Pi^k$) to estimate the displacement of $\pi^k_i$ for $i \in A_{\text{unmatched}}$. 
Depending on whether the $B^{\ast}_{\epsilon}(\pi^k_i)$-particles  have matched pairs in frame $k-1$,
we make the further decomposition 
$A_{\text{unmatched}} = A_{\text{unmatched}, e}\cup A_{\text{unmatched}, n}$.
The subset $A_{\text{unmatched}, e}$ is effective and  denotes the IDs of particles  that have $B^{\ast}_{\epsilon}(\pi^k_i)$-particles with matched pairs at frame $k-1$.
The set $A_{\text{unmatched}, n}$ is the complement $A_{\text{unmatched}} \setminus A_{\text{unmatched}, e}$ and denotes the IDs of particles whose $B^{\ast}_{\epsilon}(\pi^k_i)$-particles have no matched pairs at frame $K-1$. 
Mathematically, we write down
\begin{align*}
    A_{\text{unmatched}, e} = 
    \left\{
     i \in A_{\text{unmatched}} \
     \middle
     \vert
     \
     \exists j\in A_{\text{matched}}, 
     \text{such that}
     \
     \pi^k_j \in
     B^{\ast}_{\epsilon}(\pi^k_i)
    \right\}.
\end{align*}
Next, for $i\in A_{\text{unmatched},e}$ we denote the set
\begin{align*}
    A_{\text{matched}, ei}
    =
    \left\{
    j\in A_{\text{matched}}
    \
    \middle\vert
    \
    \pi^k_j \in 
    B^{\ast}_{\epsilon}(\pi^k_i)
    \right\},
\end{align*}
from which we use the weighted average displacement to estimate the displacement of the particle $\pi^k_{i}$ in $A_{\text{unmatched}, e}$.
For the particles with IDs in $A_{\text{unmatched}, n}$, since all their $B^{\ast}_{\epsilon}(\pi^k_i)$-particles are not paired to particles in frame $k-1$, 
we keep the zero-order prediction method.
Fig.~\ref{set partition} illustrates how we partition the particles.
\begin{figure}[!h]
  \centering
\includegraphics[scale=0.45]{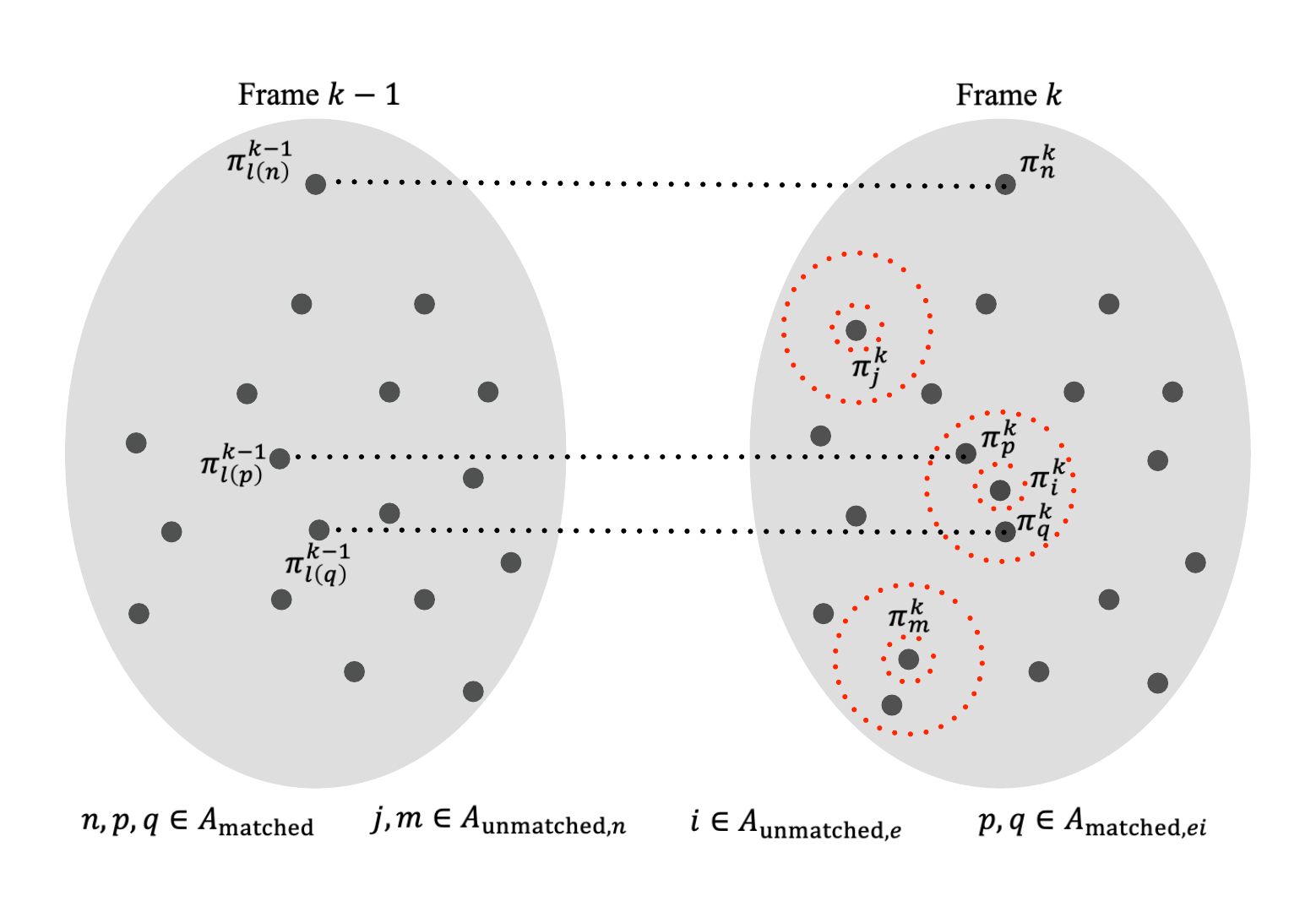}
  \caption{Visualization of set partition for the first-order prediction.}
  \label{set partition}
\end{figure}

\paragraph{Case 1: deterministic reconstruction.}
For the deterministic reconstruction case, the transport cost is
\begin{align*}
    c_{ij}=W_{2}
        \left(
        \delta_{r^{k+1|k}_i},
        \delta_{r^{k+1}_j}
        \right)^2,
\end{align*}
where the notation $ \delta_{r^{k+1|k}_i}$ means the predicted particle position from frame $k$ to $k+1$ for the $i$th particle.
Since we decompose the particle IDs into three subsets 
$
A=A_{\text{matched}}\cup
A_{\text{unmatched}, e}\cup
A_{\text{unmatched}, n}
$, we use three different strategies to make the position prediction. 
For $i\in A_{\text{matched}}$, we simply linearly extend the particle position.
For $i\in A_{\text{unmatched}, e}$, we use Wasserstein distance weighted predicted displacements of $B^{\ast}_{\epsilon}(\pi^k_i)$-particles to estimate the predicted particle position.
For $i\in A_{\text{unmatched}, n}$, we keep the zero-order prediction.
We assume the time steps between recorded images are the same.
Then, the predicted particle position formula is summarized as follows:
\begin{align*}
    r^{k+1|k}_i = 
        \begin{cases}
            r^k_i + \left(
            r^k_i -r^{k-1}_{l(i)}
            \right),
            \
            &\forall i \in A_{\text{matched}},
            \\
            r^k_i
            +
             \dfrac{
             \sum_{j\in A_{\text{matched}, ei}}
            \left(
            r^k_j - r^k_i
            \right)
            \left(
            r^k_j -r^{k-1}_{l(j)}
            \right)
            }{
            \sum_{j\in A_{\text{matched}, ei}}
            \left(
            r^k_j - r^k_i
            \right)
            },
            \
            &\forall i \in A_{\text{unmatched}, e},
            \\
            r^k_i, 
            &\forall i \in A_{\text{unmatched}, n}.
        \end{cases}
\end{align*}

\paragraph{Case 2: stochastic reconstruction.}
For the Gaussian stochastic reconstruction case, the transport cost is
\begin{align*}
    c_{ij} = W_{2}
        \left(
        \mathcal{N}(m^{k+1|k}_i, \Sigma^{k+1|k}_i),\
    \mathcal{N}(m^{k+1}_j, \Sigma^{k+1}_j)
     \right)^2,
\end{align*}
where $m^{k+1|k}_i$ and $\Sigma^{k+1|k}_i$ are the predicted mean and covariance matrix of the Gaussian distribution.
Since BVR uses diagonal multivariate normal distributions as the guide, we present a first-order method when the covariance matrices are diagonal and only implemented this parsimonious case numerically.
The general case when the covariance matrices are dense significantly increases the computational costs.
We leave the discussion of the full covariance matrix case for Appendix \ref{appendix: full cov}.
We denote $R = (X, Y, Z)$ as the random position vector and $\mathbb{V}$ as the variance operator.
Similar to case 1, three subcases need to be considered.

For $i\in A_{\text{matched}}$, 
linear extension has the form
\begin{align*}
    R^{k+1|k}_i = R^k_i + 
    \left(
    R^k_i - R^{k-1}_{l(i)}
    \right)
\end{align*}
where
\begin{align*}
    R^k_i\sim \mathcal{N}
    \left(
    m^{k}_{i}, 
    \operatorname{diag}\left\{
    \left(\sigma^k_{i, X}\right)^2,
    \left(\sigma^k_{i, Y}\right)^2,
    \left(\sigma^k_{i, Z}\right)^2
    \right\}
    \right)
\end{align*}
and
\begin{align*}
    R^{k-1}_{l(i)} \sim 
    \mathcal{N}
    \left(m^{k-1}_{l(i)}, 
    \operatorname{diag}\left\{
    \left(\sigma^{k-1}_{l(i), X}\right)^2,
    \left(\sigma^{k-1}_{l(i), Y}\right)^2,
    \left(\sigma^{k-1}_{l(i), Z}\right)^2
    \right\}
    \right).
\end{align*}

We then make a valid simplification that assumes $R^k_i$ and $R^{k-1}_{l(i)}$ are independent. 
So the predicted Gaussian distribution is
\begin{align*}
    R^{k+1|k}_i
    \sim
    \mathcal{N}
    \left(m^{k+1|k}_{i}, 
    \operatorname{diag}\left\{
    \left(\sigma^{k+1|k}_{i, X}\right)^2,
    \left(\sigma^{k+1|k}_{i, Y}\right)^2,
    \left(\sigma^{k+1|k}_{i, Z}\right)^2
    \right\}
    \right)
\end{align*}
 where the predicted mean and standard deviations can be computed by 
\begin{align*}
    m^{k+1|k}_{i} &= 
    \mathbb{E}
    \left[
    2R^{k}_i- R^{k-1}_{l(i)}
    \right]
    =
    2m^k_i - m^{k-1}_{l(i)}
\end{align*}
and
\begin{align*}
    (\sigma^{k+1|k}_{i, \psi})^2
    =\mathbb{V}
    \left[
    2\psi^{k}_{i}- \psi^{k-1}_{l(i)}
    \right]
    =
     4\left(\sigma^k_{i, \psi}\right)^2 + \left(\sigma^{k-1}_{l(i), \psi}\right)^2, \
     \text{for}, 
     \psi = X, Y, \text{and}\ Z.
\end{align*}

For $i\in A_{\text{unmatched}, e}$, we have to estimate the translation and stretches along $x$, $y$ and $z$ axes of the Gaussian ellipsoid using $B^{\ast}_{\epsilon}(\pi)$-particles.
For the translation, we use Wasserstein distance weighted predicted displacements of $B^{\ast}_{\epsilon}(\pi^k_i)$-particles to estimate the predicted mean position:
\begin{align*}
    m^{k+1|k}_i
    =
    m^k_i + 
     \dfrac{
     \sum_{j\in A_{\text{matched}, ei}}
    W_{2}\left(
    \mathcal{N}(m^k_j, \Sigma^k_j),
    \
    \mathcal{N}(m^{k}_i, \Sigma^{k}_i)
    \right)
    \left(
    m^k_j -m^{k-1}_{l(j)}
    \right)
    }{
     \sum_{j\in A_{\text{matched}, ei}}
   W_{2}\left(
    \mathcal{N}(m^k_j, \Sigma^k_j),
    \
    \mathcal{N}(m^{k}_i, \Sigma^{k}_i)
    \right)
    }.
\end{align*}
For the stretches, we use weighted stretch ratios:
\begin{align*}
    \sigma^{k+1|k}_{i, \psi}
    =
    \sigma^{k}_{i, \psi}
    \dfrac{
     \sum_{j\in A_{\text{matched}, ei}}
    W_{2}\left(
    \mathcal{N}(m^k_j, \Sigma^k_j),
    \
    \mathcal{N}(m^{k}_i, \Sigma^{k}_i)
    \right)
    \left(
    \frac{\sigma^{k+1|k}_{j,\psi}}{\sigma^{k}_{j,\psi}}
    \right)
    }{
     \sum_{j\in A_{\text{matched}, ei}}
   W_{2}\left(
    \mathcal{N}(m^k_j, \Sigma^k_j),
    \
    \mathcal{N}(m^{k}_i, \Sigma^{k}_i)
    \right)
    }, \quad
    \text{for} \
    \psi = X, Y,  \text{and}\ Z.
\end{align*}

For $i\in A_{\text{unmatched}, n}$, we keep the zero-order prediction.

In summary, for the stochastic reconstruction using diagonal covariance matrices, we have
\begin{align*}
    m^{k+1|k}_i
    =
    \begin{cases}
        m^{k}_i + \left(
        m^k_i - m^{k-1}_{l(i)}
        \right), &\forall i \in A_{\text{matched}},
        \\
        m^k_i + 
         \dfrac{
         \sum_{j\in A_{\text{matched}, ei}}
        W_{2}\left(
        \mathcal{N}(m^k_j, \Sigma^k_j),
        \
        \mathcal{N}(m^{k}_i, \Sigma^{k}_i)
        \right)
        \left(
        m^k_j -m^{k-1}_{l(j)}
        \right)
        }{
         \sum_{j\in A_{\text{matched}, ei}}
       W_{2}\left(
        \mathcal{N}(m^k_j, \Sigma^k_j),
        \
        \mathcal{N}(m^{k}_i, \Sigma^{k}_i)
        \right)
        }, & \forall i \in A_{\text{unmatched}, e},
        \\
        m^k_i, & \forall i \in A_{\text{unmatched}, n},
    \end{cases}
\end{align*}
and for $\psi=X, Y$ and $Z$
\begin{align*}
    \sigma^{k+1|k}_{i, \psi}
    =
    \begin{cases}
        \left(
        4\left(\sigma^{k}_{i, \psi}\right)^2
        + 
        \left(
            \sigma^{k-1}_{l(i), \psi}
        \right)^2
        \right)^{\frac{1}{2}}
        , &\forall i \in A_{\text{matched}},
        \\
        \sigma^{k}_{i, \psi}
        \dfrac{
         \sum_{j\in A_{\text{matched}, ei}}
        W_{2}\left(
        \mathcal{N}(m^k_j, \Sigma^k_j),
        \
        \mathcal{N}(m^{k}_i, \Sigma^{k}_i)
        \right)
        \left(
        \frac{\sigma^{k+1|k}_{j,\psi}}{\sigma^{k}_{j,\psi}}
        \right)
        }{
         \sum_{j\in A_{\text{matched}, ei}}
       W_{2}\left(
        \mathcal{N}(m^k_j, \Sigma^k_j),
        \
        \mathcal{N}(m^{k}_i, \Sigma^{k}_i)
        \right)}
        , & \forall i \in A_{\text{unmatched}, e},
        \\
        \sigma^{k}_{i, \psi}, & \forall i \in A_{\text{unmatched}, n}.
    \end{cases}
\end{align*}
The first-order prediction method relies on faithfully estimating the predicted displacement.
The estimation error will propagate forward in frames, which eventually degrades the performance of the first-order prediction method.
We will discuss how to remove corrupted displacement predictions in section \ref{section_hyper}.

\subsubsection{Interpretation of the optimal transport plan and the effect of hyperparameter transport number}
\label{section_interpret}

In this section, we give an interpretation of the optimal transport plan $\gamma^{\ast}$ and discuss how the hyperparameter transport number $N_p$ affects $\gamma^{\ast}$.

Fig. \ref{otplan} plots a case of optimal transport matrix $\gamma^{\ast}$ after solving PWD problem in Eqs.~(\ref{uot}) for an illustrated 2D case where $N=100$, $M=110$, and $N_p=85$.
Since the matrix $\gamma^{\ast}$ only contains 0 and 1.
This allows us to match the $i$th reconstructed particle position $\pi^k_i$ at frame $k$  to the $j$th reconstructed particle position $\pi^{k+1}_j$ at frame $k+1$  if $\gamma^{\ast}_{ij}=1$.
It should be emphasized that since the total transport number $N_p$ of particles we intend to match is less than $N$ and $M$, there are several zero rows and columns in $\gamma^{\ast}$.
By definition, when the $i$th row $\gamma_{i:}$ is zero, the $i$th particle  $\pi^k_i$ will not be transported to any target in $\pi^{k+1}$. 
In such a case, we say a particle track ends at frame $k$.
This may happen because of reconstruction errors, particles flowing out of the flow domain, or matching algorithm errors.
Similarly, when jth column $\gamma_{:j}$ is 0, there is no particle at frame $k$ transported to the $j$th particle at frame $k+1$.
Then we create a new particle track starting with $\pi^{k+1}_j$ at frame $k+1$. 

\begin{figure}[!h]
  \centering
\includegraphics[scale=0.50]{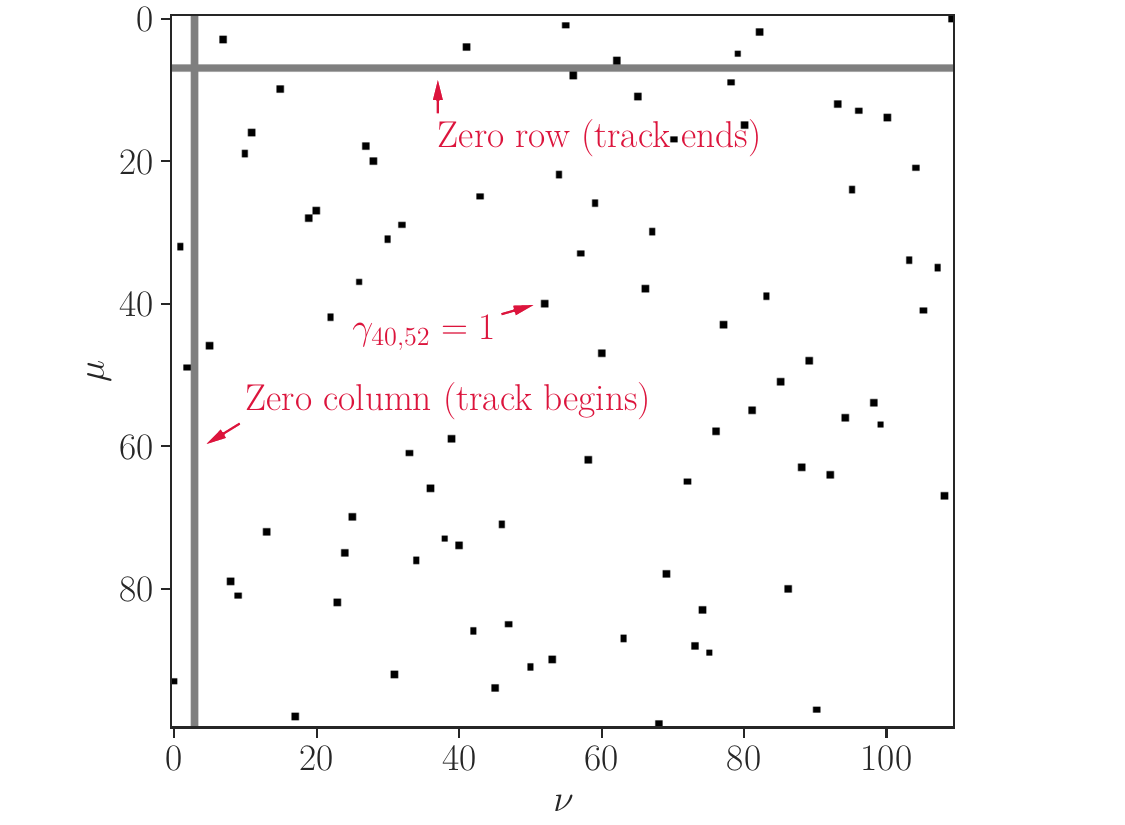}
  \caption{Optimal transport matrix ($N$=100, $M=110$ and $N_p=85$).}
  \label{otplan}
\end{figure}

The IUOT matching algorithm is robust since it can automatically detect suspicious particle pairs that have large displacements between two frames, and associate either zero row or column to them.
Also, it prioritizes matching particles that have small displacements and assigns 1 to the transport plan matrix entry $\gamma_{ij}^{\ast}$.

However, we should remind the readers that the number of zero rows and columns is not only determined by the number of suspicious particle pairs but also the hyperparameter transport number.
We set  $N_{p, \alpha} = \operatorname{ceil}\left(\alpha \min\{N, M\}\right)$, where $0<\alpha\le 1$.
 Fig. \ref{Np} shows the matching results for the previous 2D case with varying $\alpha$.
 When $\alpha=1$, we observe several large particle jumps.
 This is because we are overconfident and the total number of particles we should match is less than what we assume.
 As we decrease the value of $\alpha$, mismatches decrease.
 However, when $\alpha$ is too small, e.g., $\alpha=0.80$, we are too conservative and the number of particles we have matched is less than what we should match.
 An obvious better choice of $\alpha$ is 0.85 from observation.
 Therefore, selecting this hyperparameter is crucial, and we should automatically decide the value of $\alpha$ through hyperparameter optimization.
\begin{figure}[!h]
  \centering
\includegraphics[scale=0.42]{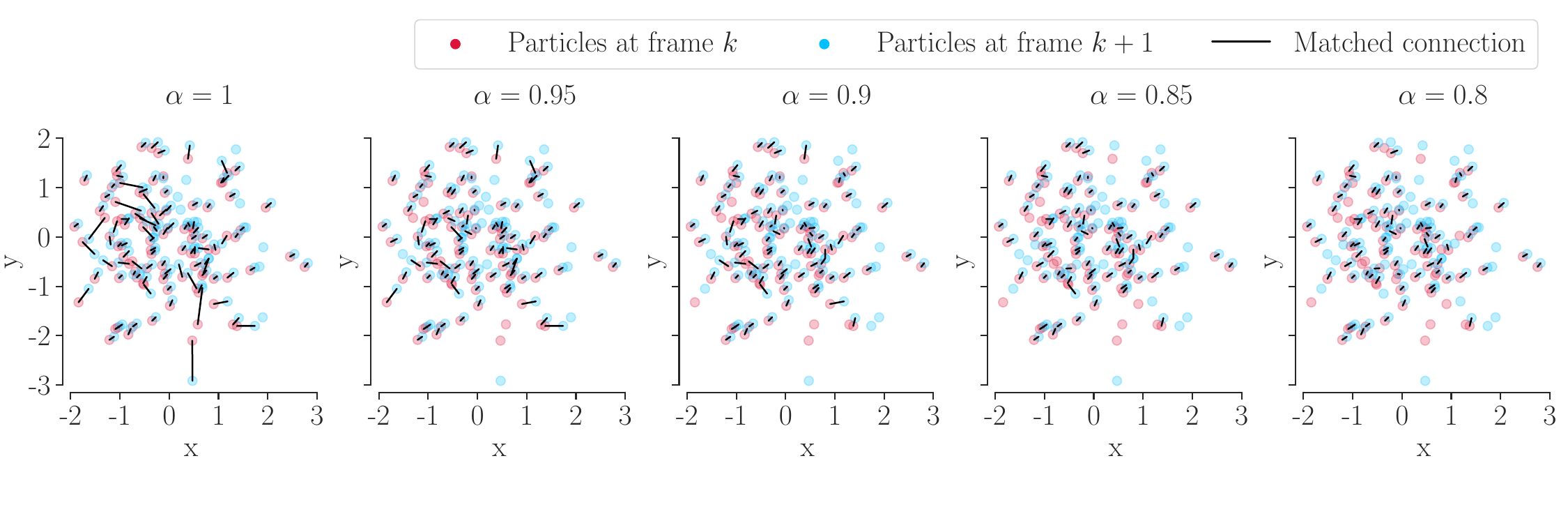}
  \caption{Particle matching with varying $\alpha$.}
  \label{Np}
\end{figure}

\subsubsection{Optimizing the hyperparameter transport number}
\label{section_hyper}

In this section, we optimize the hyperparameter $\alpha$.
To evaluate the matching performance, we use the yield rate $E_y$ and reliability rate $E_r$ defined in \cite{malik1993particle} and \cite{cardwell2011multi}:
\begin{align}
    E_y = \frac{N_c}{N_g},
    \
    \text{and}
    \
    E_r = \frac{N_c}{N_{p,\alpha}}.
    \label{metric}
\end{align}
In the above evaluation metrics, $N_c$, $N_g$, and $N_{p,\alpha}$ denote the numbers of correct matches, ground-truth matches and algorithm-extracted matches, respectively.
\begin{remark}
 The optimal $\alpha$ should achieve the highest number of correct matches $N_c$, hence the highest yield rate $E_y$,  and keep the reliability rate  $E_r$ at 1.   
\end{remark}
Namely, we want the maximize $N_{p, \alpha}$ while $E_r=1$.
The optimal $\alpha$ apparently exists since the number of particles is finite.
The big issue is we only know the value of $N_{p, \alpha}$, which is specified by us as the hyperparameter.
\begin{remark}
   The goal of hyperparameter optimization is to maximize an estimated number of correct matches $\hat{N}_c$, hence an estimated yield rate $\hat{E}_y$,  while keeping an estimated reliability rate $\hat{E}_r$ higher than a user-specified threshold. 
\end{remark}
This is achieved by comparing each particle displacement with the neighborhood particle displacements.

The overall hyperparameter optimization strategy follows: 
1) discretize the range of $\alpha$, 
2) solve the optimal transport plan for each $\alpha$ in parallel, 
3) identify faithful pairs and compute an estimated correct matches $\hat{N}_c$,
4) compute an estimated yield rate $\hat{E}_r$ to approximate the unknown yield rate $E_r$, 
5) reject $\alpha$ if $\hat{E}_r$ is below a specified threshold,
and 
6) remove non-faithful pairs.
A flow chart to visualize these six steps is shown in Fig. \ref{hyper_optimization}.

\begin{figure}[!htb]
  \centering
\includegraphics[scale=0.20]{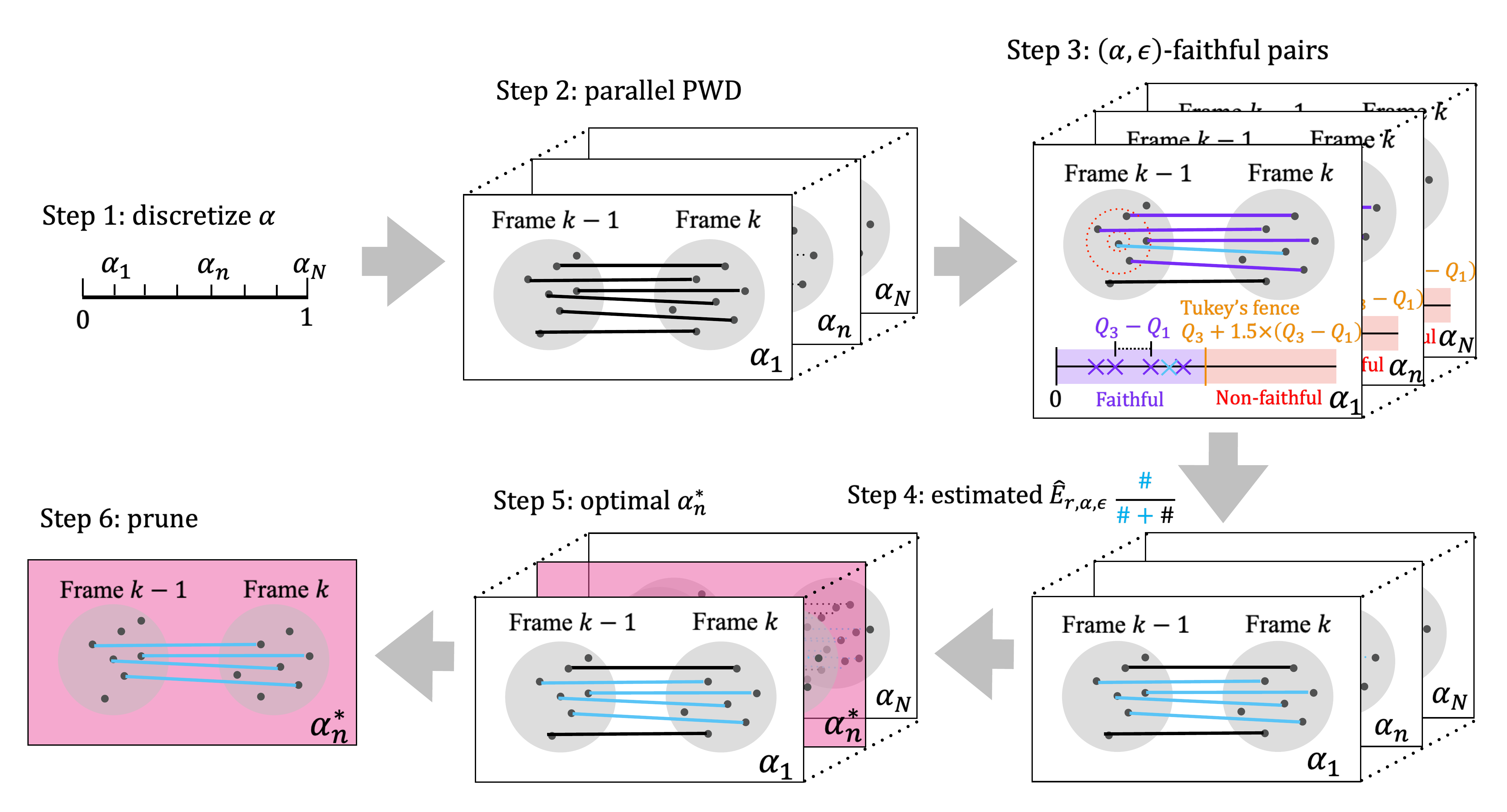}
  \caption{A flow chart of optimizing the relative transport number $\alpha$.}
  \label{hyper_optimization}
\end{figure}

The first step is to discretize the range $(0, 1]$ of $\alpha$ into a finite set of numbers
$0< \alpha_1 <,\cdots, <\alpha_N\le 1$.
In the second step,  we parallelly solve the PWD problem in \qref{uot} for each $\alpha_n$ to get a set of optimal transport matrix
$\left\{\gamma^{\ast}(\alpha_1),\cdots, \gamma^{\ast}(\alpha_N) \right\}$.
The third step is to find the number of faithful particle pairs identified in each $\gamma^{\ast}(\alpha_n)$.
To this end, we compare the displacement of each identified pair $\left(\pi^k_i, \pi^{k+1}_j\right)$ with the empirical displacement distribution of the $B^{\ast}_{\epsilon}(\pi^k_i)$-particles.
This step aims to check if the displacement of the particle $\pi^k_i$ is significantly larger than the displacements of the neighborhood particles.
The empirical distribution is represented by the finite set
\begin{align*}
    \widehat{\Delta}_{\alpha,\epsilon}\left(\pi^k_i\right)
    =
    \left\{
    W_2\left(
    \pi^k_{p}, \pi^{k+1}_q)
    \right)
    \
    \middle
    \vert
    \
    (p, q)\ 
    \text{such that}\
    \pi^k_p \in B^{\ast}_{\epsilon}\left(\pi^k_i \right)
    \
    \text{and}
    \
    \gamma^{\ast}_{pq}(\alpha)=1
    \right\}
\end{align*}
Then, we use Tukey's fences \cite{tukey1977exploratory} to decide whether to reject this pair.
Specifically, the statistical comparison criterion uses the first 
$
Q1\left(\widehat{\Delta}_{\alpha, \epsilon}\left(\pi^k_i \right)\right)
$ and third 
$
Q3\left(\widehat{\Delta}_{\alpha, \epsilon}\left(\pi^k_i \right)\right)
$ quartiles.
We define the faithful particle pair:
\begin{definition}[$(\alpha, \epsilon)$-faithful particle pair]
    An $(\alpha, \epsilon)$-faithful particle pair $(\pi^{k}_i, \pi^{k+1}_j)$ is an identified particle pair in $\gamma^{\ast}(\alpha)$, i.e., $\gamma^{\ast}_{ij}(\alpha)=1$, which satisfies
    \begin{align*}
        W_2(\pi^k_i, \pi^{k+1}_j)
        \le
        Q3\left(
        \widehat{\Delta}_{\alpha, \epsilon}\left(\pi^k_i \right)\right) 
        + 
        1.5 \times\left(
        Q3\left(
        \widehat{\Delta}_{\alpha, \epsilon}\left(\pi^k_i \right)\right) 
        - 
        Q1
        \left(
        \widehat{\Delta}_{\alpha, \epsilon}\left(\pi^k_i \right)\right) 
        \right).
    \end{align*}
\end{definition}
We test whether each identified pair is faithful or not and denote $\hat{N}_{c, \alpha, \epsilon}$ the total number of $(\alpha, \epsilon)$-faithful pairs.
We use $\hat{N}_{c, \alpha, \epsilon}$  to approximate the unknown number of correct matches $N_{c}$.

In step 4, we compute the estimated reliability rate $\hat{E}_{r, \alpha, \epsilon}$ defined by
\begin{align*}
    \hat{E}_{r, \alpha, \epsilon}
    =
    \frac{
    \hat{N}_{c, \alpha, \epsilon}
    }{N_{p, \alpha}}.
\end{align*}
After computing the estimated reliability rate for each $\alpha_n$, in step 5, we define the threshold for each $\alpha$
\begin{align*}
    \left(
    \eta_{1},
    \eta_2,
    \cdots,
    \eta_{N-1},
    \eta_{N}
    \right)
    =
    \left(
    0, 
    \frac{N_{p,\alpha_1}}{N_{p, \alpha_2}}
    ,
    \cdots,
    \frac{N_{p, \alpha_{N-2}}}{N_{p, \alpha_{N-1}}},
    \frac{N_{p, \alpha_{N-1}}}{N_{p, \alpha_N}}
    \right),
\end{align*}
and reject $\alpha_n$ if $\hat{E}_{r, \alpha_n, \epsilon}< \eta_n$.
Notice that we always accept the smallest $\alpha_1$.
The optimal $\alpha^{\ast}$ is the largest accepted $\alpha_n$ denoted by $\alpha_n^{\ast}$.

The last step is to remove the non-faithful pairs identified in $\gamma^{\ast}(\alpha^{\ast}_n)$. 
This pruning step is essential for the first-order prediction method since it is undesirable to make displacement predictions using non-faithful pairs and this error will propagate forward in frames.
Our synthetic experiments showed that this significantly improves the matching performance of the first-order prediction approach when facing large position reconstruction errors.

Algorithm \ref{alg:hyper} summarizes the essential steps to robustly match reconstructed particles using PWD with hyperparameter optimization.

\RestyleAlgo{ruled} 
\SetKwInput{KwInit}{Initialization}
\SetKwInput{KwInput}{Input}
\SetKwInput{KwReturn}{Return}
\SetKwInput{Kwalg}{Robust\_PWD\_HyperOptimal}
\SetKwComment{Comment}{$\triangleright$}{}
\begin{algorithm}[!ht]
\caption{
Robust PWD particle matching with hyperparameter optimization.
}\label{alg:hyper}
\Kwalg{
(
\begin{tabbing}
\hspace{1em} 
\= 
$(\mu, \nu)$, \hspace{5em} 
\= \quad
 \# \textit{Source and target.}
\\
\> 
$\left\{
\alpha_1, \cdots, \alpha_N
\right\}$, 
\> \quad
\# \textit{Discretized $\alpha$.}
\\
\>
$\epsilon$,
\> \quad
\# \textit{$\epsilon$-ball radius.}
\\
)
\end{tabbing}
}
\# Do in parallel
\;
 \For{$\alpha_n$ in $\left\{
\alpha_1, \cdots, \alpha_N
\right\}$ }{  
 1.
    Compute the optimal transport plan $\gamma^{\ast}(\alpha_n)$ and the estimated reliability rate $\hat{E}_{r, \alpha_n, \epsilon}$.
 \\
 2.
  Check if 
  $
  \hat{E}_{r, \alpha_n, \epsilon}< \eta_n
  $ to reject $\alpha_n$.
}
  3.
  Find the largest accepted $\alpha_n^{\ast}$ 
  and remove the
  $N_{p, \alpha^{\ast}_n}-\hat{N}_{r, \alpha^{\ast}_n, \epsilon}$ 
  number of non-faithful pairs from $\gamma^{\ast}(\alpha^{\ast}_n)$ to get a clean version of the optimal relative transport number 
  $\alpha^{\ast}=\hat{N}_{r, \alpha^{\ast}_n, \epsilon}/ \min\{N, M\}$ and
  the optimal transport plan $\gamma^{\ast}(\alpha^{\ast})$. 
  \\
 \KwReturn{
 $
 \left(
 \alpha^{\ast},
 \gamma^{\ast}(\alpha^{\ast})
 \right)
 $
 \hspace{1.5em}
 \#
 (\textit{Optimal relative transport number, 
 Optimal transport plan}).
 }
\end{algorithm}

\subsection{Extraction of particle tracks from matched pairs}

We explain how to patch the matched particles between two consecutive time frames to reconstruct whole particle tracks.
We assume that reconstructed particle positions at all $K$ time frames are provided 
$
\Pi = \left(
\Pi^1, \cdots, \Pi^K
\right)
$.
Each $\Pi^k = \left(
\pi^k_1, \cdots, \pi^k_{M^k}
\right)$ has $M^k$ number of particle position estimates, and we call $A^k = \left\{ 1, \cdots, M^k\right\}$ the set of particle IDs at frame k.
We also have solved the PWD problem in Eqs.~(\ref{uot}) $K-1$ times to get a sequence of optimal transport matrices
$
Q=
\left(
    \gamma^{\ast 1}, \cdots, \gamma^{\ast K-1}
\right)
$, where each $\gamma^{\ast k}$ is an $M^{k}$ by $M^{k+1}$ matrix.
The objective is to extract all particle tracks from the optimal transport matrices.

First, we make some definitions and notations.
Let 
$\left\{k\right\} = \left(1, \cdots, K\right)$
denote the sequence of time frames, $\left\{ k_i \right\}$
to denote a subsequence of $\left\{k\right\}$.
Then we reconstruct $N_T$ particle tracks with the track ID set $B=\left\{1, \cdots, N_T\right\}$.
For the $q$th track, we associate a subsequence $\left\{k^q_i\right\}$, where $i\in \left\{1, \cdots, I^q\right\}$. This means the subsequence $\left\{k^q_i\right\}$ has a length of $I^q$.
Moreover, the subsequence $\left\{k^q_i\right\}$ is an arithmetic sequence with a common difference of 1, i.e., 
$
k^q_{i+1} = k^q_{i}+1, \forall k\in \left\{1, \cdots, I^q-1\right\}
$.
Next, we define a function $f\left(q, k^q_i\right)$ that maps the track ID $q$ and the time frame number $k^q_i$
to a particle ID in the set $A^{k^q_i}$ at frame $k^q_i$.
Recall $\pi^k_i$ denotes the $i$th reconstructed particle position at time frame $k$,
then we make the formal definition,
\begin{definition}[Proper tracks]
Let $\Pi$ be a sequence of reconstructed particle positions.
We define the $q$th particle track: 
\begin{align*}
   \chi_q = \left\{
\pi^{k^q_i}_{f\left(p, k^q_i\right)}
\right\},\ \text{where} \
i\in \left\{1, \cdots, I^q\right\}.
\end{align*}
Then a proper set of particle tracks $\Phi=\left\{\chi_q\right\}_{q\in B}$ is a partition of $\Pi$, i.e., $\Pi$ is the disjoint union of $\chi_q$:
\begin{align*}
    \biguplus_{q}\chi_q = \Pi.
\end{align*}
Namely, the set $\Phi$ is exhaust and pairwise disjoint.
\end{definition}

\begin{definition}[Optimal tracks]
    \label{def2}
    An optimal set of particle tracks $\Phi^{\ast}$ is a proper set of particle tracks retrieved using the set $Q$ of optimal transport plan matrices that satisfies:
    \begin{enumerate}
        [label=(2.\arabic*)]
        \item \textnormal{(Initialization condition)} \ $\pi^1=\left\{
        \pi^{k^q_1}_{f\left(q, k^q_1\right)}\middle\vert q\in B \ \text{and} \ k^q_1=1
        \right\}$.\label{21}
        \item 
        \textnormal{(Matching condition)} \
        $\gamma^{\ast k^q_i}_{nm} = 1$, where $n=f\left(q, k^q_i\right)$ and $m=f\left(q, k^q_{i+1}\right)$, for all $q\in B$ and $i\in \left\{1,\cdots, I^{q}-1\right\}$,
        where $I^q>1$.\label{22}
        \item
        \textnormal{(New track start condition)} \
        $\gamma^{\ast k^q_1-1}_{:m}=0$, where $m=f\left(q, k^q_1\right)$, for all $q\in B$ such that $k^p_1>1$.\label{23}
        \item
        \textnormal{(Track end condition)}
            \begin{enumerate}[label=(2.4.\alph*)]
            \item 
            $\gamma^{\ast k^q_{I^q}}_{n:}=0$, where $n=f\left(q, k^q_{I^q}\right)$, for all $q\in B$ such that $k^q_{I^q}<K$.\label{24a}
            \\
            \item or $k^q_{I^q}=K$.\label{24b}
            \end{enumerate}
    \end{enumerate}
\end{definition}

A graphic illustration of definition \ref{def2} is shown in Fig. \ref{graph}.
Specifically, the figure explains three cases of keeping, stopping, and creating tracks depending on the optimal transport map values.
\begin{figure}[!h]
  \centering
\includegraphics[scale=0.6]{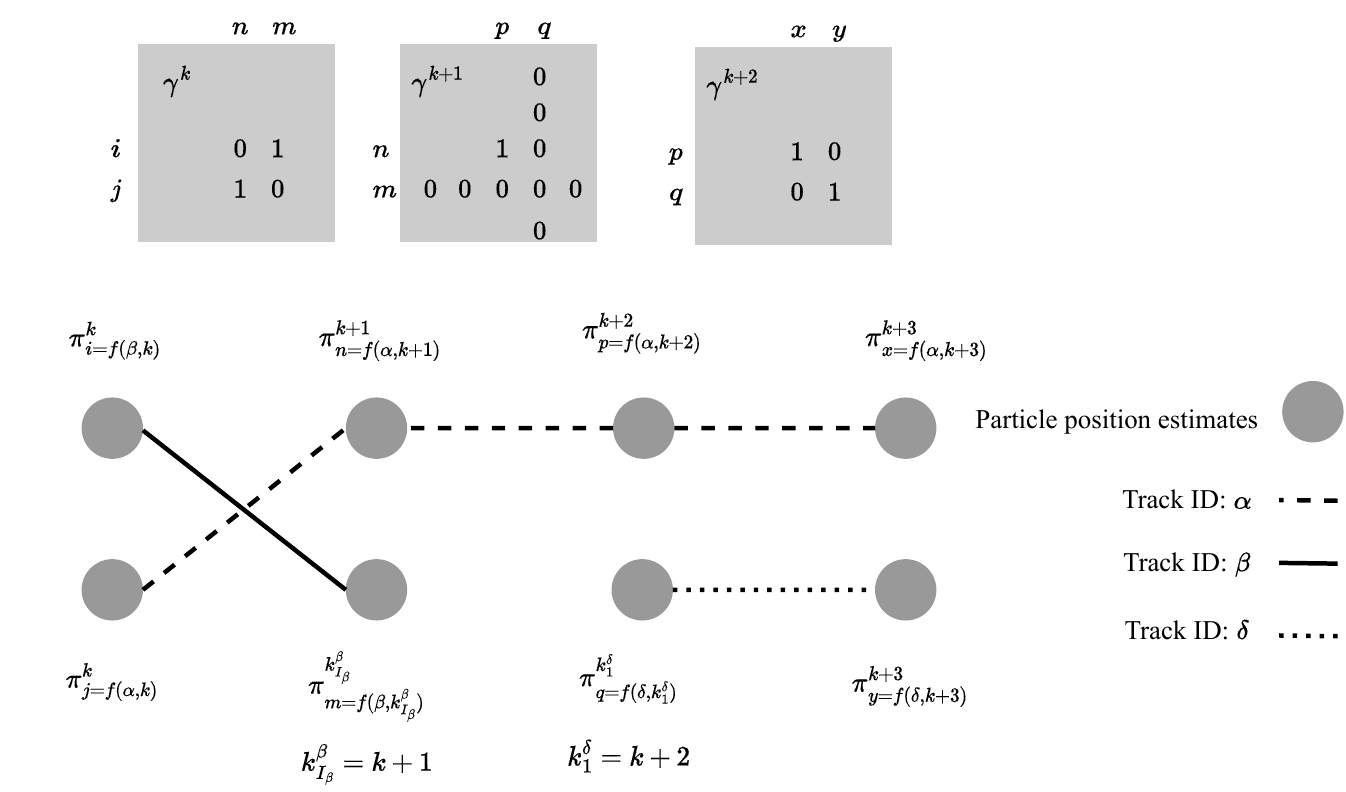}
  \caption{Reconstruction of optimal tracks from optimal transport maps.}
  \label{graph}
\end{figure}

Algorithm \ref{alg:ottrack} explains how we extract the optimal track $\Phi^{\ast}$.
\RestyleAlgo{ruled} 
\SetKwInput{KwInit}{Initialization}
\SetKwInput{KwInput}{Input}
\SetKwInput{KwReturn}{Return}
\SetKwInput{Kwalg}{PWD\_Tracks}
\SetKwComment{Comment}{$\triangleright$}{}
\begin{algorithm}[!ht]
\caption{
Optimal particle tracks.
}\label{alg:ottrack}
\Kwalg{
(
\begin{tabbing}
\hspace{1em} 
\= 
$\Pi=\left(\Pi^1, \cdots, \Pi^K\right)$, \hspace{1em} 
\= \quad
 \# \textit{Reconstructed particle positions.}
\\
\> 
$\left\{
\alpha_1, \cdots, \alpha_N
\right\}$, 
\> \quad
\# \textit{Discretized $\alpha$.}
\\
\>
$\epsilon$,
\> \quad
\# \textit{$\epsilon$-ball radius.}
\\
)
\end{tabbing}
}
\KwInit{
\\
\hspace{1em}
$\left\{\chi_q=\pi^1_q\right\}_{q\in A^1}$ and 
$\Phi^{\ast}=\left\{\chi_1, \cdots, \chi_{M^1}\right\}$.
\hspace{1em} \# Initialize optimal tracks.
}
 \For{$k=1; k\le K-1; k=k+1$}{  
  1. Compute 
 $
 \left(
 \alpha^{\ast k},
 \gamma^{\ast k}
 \right)=
 \textbf{Robust\_PWD\_HyperOptimal
 }
 \left(
 (\Pi^k, \Pi^{k+1})
 ,
 \left\{
\alpha_1, \cdots, \alpha_N
\right\},
\epsilon
 \right).
 $
 \\
 2.
 Find all $i$ and $j$ such that $\gamma^{\ast k}_{ij}=1$.
 Then from $\Phi^{\ast}$ find tracks $\{\chi_{q}\}$ that end with $\left\{\pi^{k}_{i}\right\}$, and extend these tracks $\left\{\chi_q\right\}$ by appending $\left\{\pi^{k+1}_{j}\right\}$.
 \\
 3.
 Find all $i$ such that $\gamma^{\ast k}_{i:}=0$. Then from $\Phi^{\ast}$ find tracks $\left\{\chi_q\right\}$ that end with $\left\{
 \pi^k_i
 \right\}$, and stop tracking them.
 \\
 4.
  Find all $i$ such that $\gamma^{\ast k}_{:j}=0$. Then start new tracks $
  \left\{
  \chi_q
  \right\}
  $ that start with $
  \left\{
  \pi^{k+1}_{j}
  \right\}
  $, and add them to $\Phi^{\ast}$.
}
 \KwReturn{
 $
 \left(
 \Phi^{\ast},
 \alpha^{\ast 1:K-1}
 \right)
 $
 \hspace{2em}
 \#
 (\textit{Optimal tracks, Optimal relative transport numbers}).
 }
\end{algorithm}

\section{Examples}
\label{section:examples}

In this section, we verify and validate our stochastic matching algorithm with synthetic 3D Burgers vortex and experimental aneurysm examples.
We implemented our code in Python using Python Optimal Transport package \citep{flamary2021pot}.

\subsection{Synthetic example: stochastic reconstruction in 3D Burgers vortex}
\label{section:syn_sto}

First, we validate and evaluate our method through a synthetic 3D Burgers vortex example described by \cite{webster2015laboratory} and \cite{elmi2021response} :
\begin{align*}
    v_r &= - \beta r, \\
    v_z &= 2 \beta z, \\
    v_{\theta} &= \frac{\Gamma}{2\pi r}
    \left(1- \exp\left(-\dfrac{\beta r^2}{2\nu}\right)\right).
\end{align*}
In the above velocity field, $\beta=20 \ \text{s}^{-1}$ is the strain rate, $\Gamma = 2\ \text{m}^2 \ \text{s}^{-1}$ is the vortex circulation, and $\nu=1\times 10^{-6} \ \text{kg}\ \text{m}^{-1}\text{s}^{-1}$ is the kinematic viscosity of the fluid.

We use yield $E_y$ and reliability $E_r$ rates in \qref{metric} to assess the matching performance.
To measure the tracking difficulty, we use the ratio of the average particle spacing $\Delta_0$, to the maximum particle displacement 
$u_{\text{max}}\Delta t$
\begin{align*}
    p=\frac{\Delta_0}{u_{\text{max}}\Delta t}.
\end{align*}
The smaller $p$ is, the harder the particle matching problem is.
The average particle spacing has the formula 
$
\Delta_0 =  \left(V/N\right)^{1/d}
$,
where $V$ is the volume of the reconstruction domain,  $N$ is the number of reconstructed particles, and $d$ is the dimension of the problem (2 or 3).
To increase $p$, we fix $\Delta_0$ but increase the frame time step $\Delta t$.
In all synthetic examples, we compute $p^k$, $E^k_y$ and $E^k_r$ at each frame $k$, and report their average values.

To prepare synthetic Gaussian particle position estimates,
we randomly initialize 3000 particles in a cylinder with $\rho = 36$mm and $z\in \left[-15, 15\right]$ mm.
The reconstruction domain is a 
30mm $\times$ 30mm $\times$ 30mm cube.
Then we run Runge-Kutta ODE solver to generate 0.05-second deterministic ground-truth Lagrangian tracks with a 0.0005-second time step.
This means there are 100 frame data in total.
Fig. \ref{ground_truth Burgers} plots the three views of the generated ground-truth Lagrangian tracks.
Then, we associate a Gaussian distribution to each deterministic particle position by assigning the deterministic position as the mean and adding random standard deviations $\sigma_{X, Y, Z}$ to the three coordinates.
The random standard deviations are samples from the inverse gamma distribution (with a shape parameter $\alpha=4$) scaled by a factor of 0.0001.
To change $p$, we downsample the data using $\Delta t = 0.0005, 0.0010, 0.0020, 0.0025$ and $0.0050$ seconds, which produce averaged $p= 4.89, 2.45, 1.23, 0.98$, and $0.50$ for each case. 

\begin{figure}[!h]
 
\includegraphics[scale=0.16]{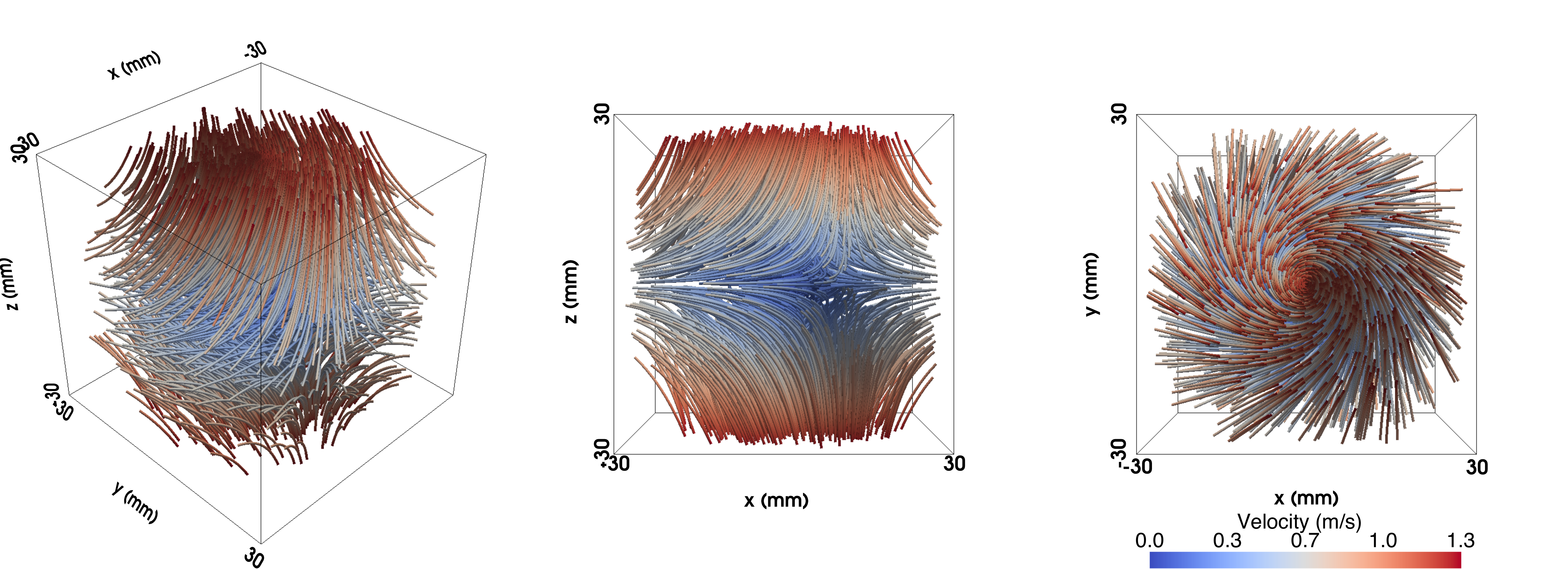}
  \caption{Ground-truth Lagrangian tracks for 3D Burgers vortex.}
  \label{ground_truth Burgers}
\end{figure}

To simulate particle position reconstruction errors, 
we gradually corrupt the aforementioned synthetic position reconstruction dataset.
The first one retains the aforementioned synthetic dataset.
This dataset is perfect (up to the standard deviations of the Gaussian estimates) because the means of Gaussian estimates coincide with the actual particle positions.
In the second case, we randomly replace $n\%$ of the perfect Gaussian estimates with $m\%$ of random Gaussian distributions in every frame data.
In this way, we introduce reconstruction errors due to missing and misidentified particles.
In the third case, we also jitter the means of remaining $(100-n)\%$ Gaussian estimates by noises.
This also considers the reconstruction noise.

\paragraph{Case 1: Perfect reconstruction.}

Fig. \ref{Burgers_nonoise_performance} shows the particle matching results given perfect reconstructed particle positions.
We observe, with a larger $\alpha$, the yield and reliability rates are higher.
The two rates are the highest when $\alpha$ is around 1.
This is, however, not the case when the reconstructed particle position estimates are not perfect.
The matching performance of the optimized $\alpha^{\ast}$ is close to the performance of $\alpha=1$.
This indicates that the hyperparameter optimization strategy successfully approximates the optimal $\alpha^{\ast}$. 
Fig. \ref{optimized_alpha_nonoise} plots the optimized $\alpha^{\ast k}$ at different time frames for the zero and first-order prediction methods under different ratios of average particle spacing to maximum particle displacements.
The ground truth $\alpha$ in the perfect reconstruction case is close to 1.
In general, the optimized relative transport numbers of the first-order prediction are greater than the zero-order prediction's.
Additionally, it is not surprising that as we increase the spacing-displacement ratio, both zero and first-order prediction algorithms achieve higher yield and reliability rates.

\begin{figure}[!h]
  \centering
\includegraphics[scale=0.45]{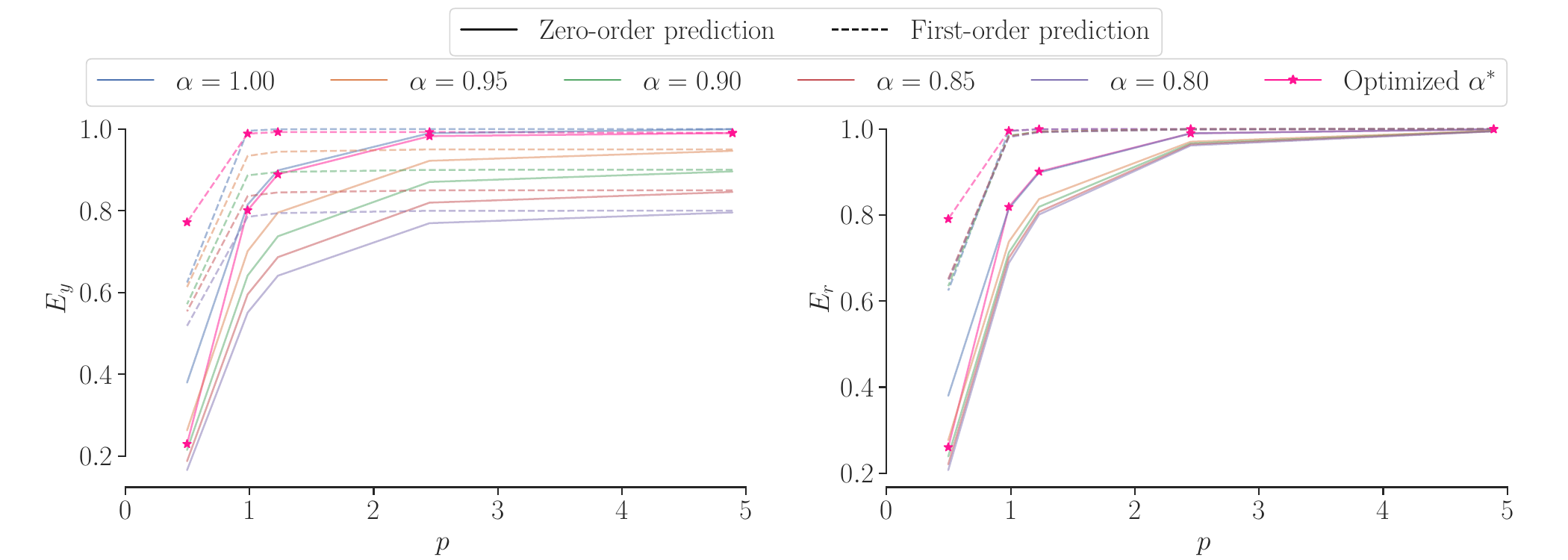}
  \caption{Case 1:
  mean yield and reliability rates for different spacing-displacement parameters $p$ and relative transport number $\alpha$.}
  \label{Burgers_nonoise_performance}

\vspace{1cm}

  \centering
\includegraphics[scale=0.44]{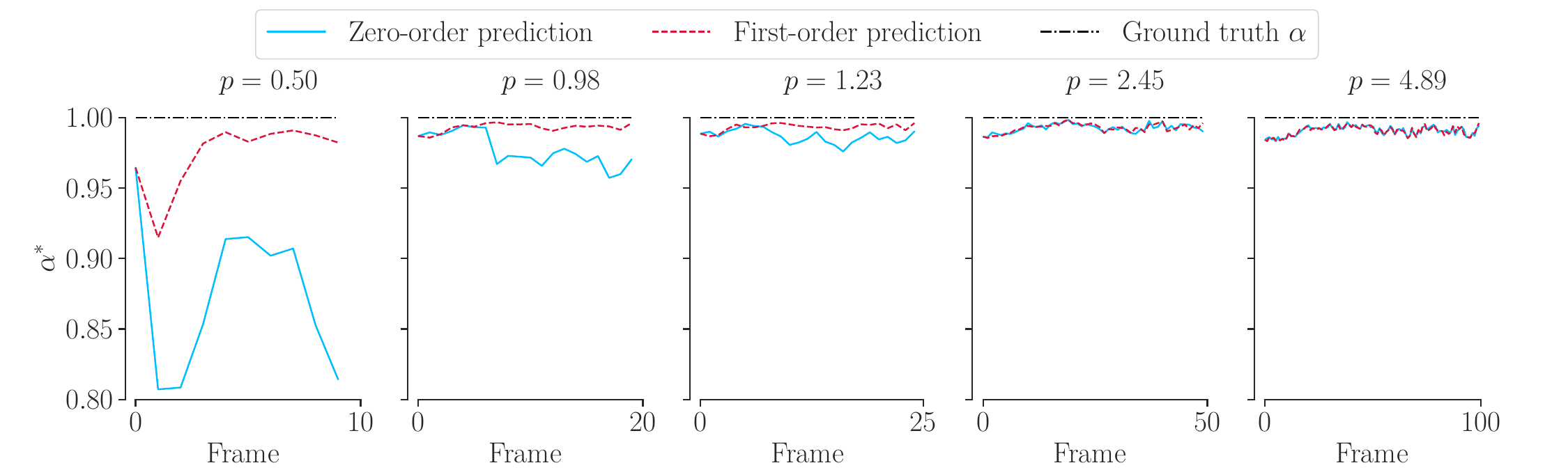}
  \caption{Case 1: optimized relative transport number $\alpha^{\ast k}$  at frame $k$ for varying $p$. }
  \label{optimized_alpha_nonoise}
\end{figure}

\paragraph{Case 2: Random replacement.}

Next, we examine the algorithm's performance after randomly replacing the perfect Gaussian particle reconstructions.
Specifically, we consider 36 different replacing levels using the pair $(n, m) \%$ in the  Cartesian product set $\left\{0, 2, 4, 6, 8, 10\right\}^2$.
Recall that the replacement level $(n, m)$ means we first randomly remove $n\%$ of perfect Gaussian distributions at each time frame.
Then, at each time frame, we add $m\%$ random Gaussian distributions whose means are drawn from a uniform distribution defined over the reconstruction volume and standard deviations are drawn from the
inverse gamma distribution (with a shape parameter $\alpha=4$) scaled by a factor of 0.0001.

\begin{figure}[!htb]
    \includegraphics[scale=0.32]{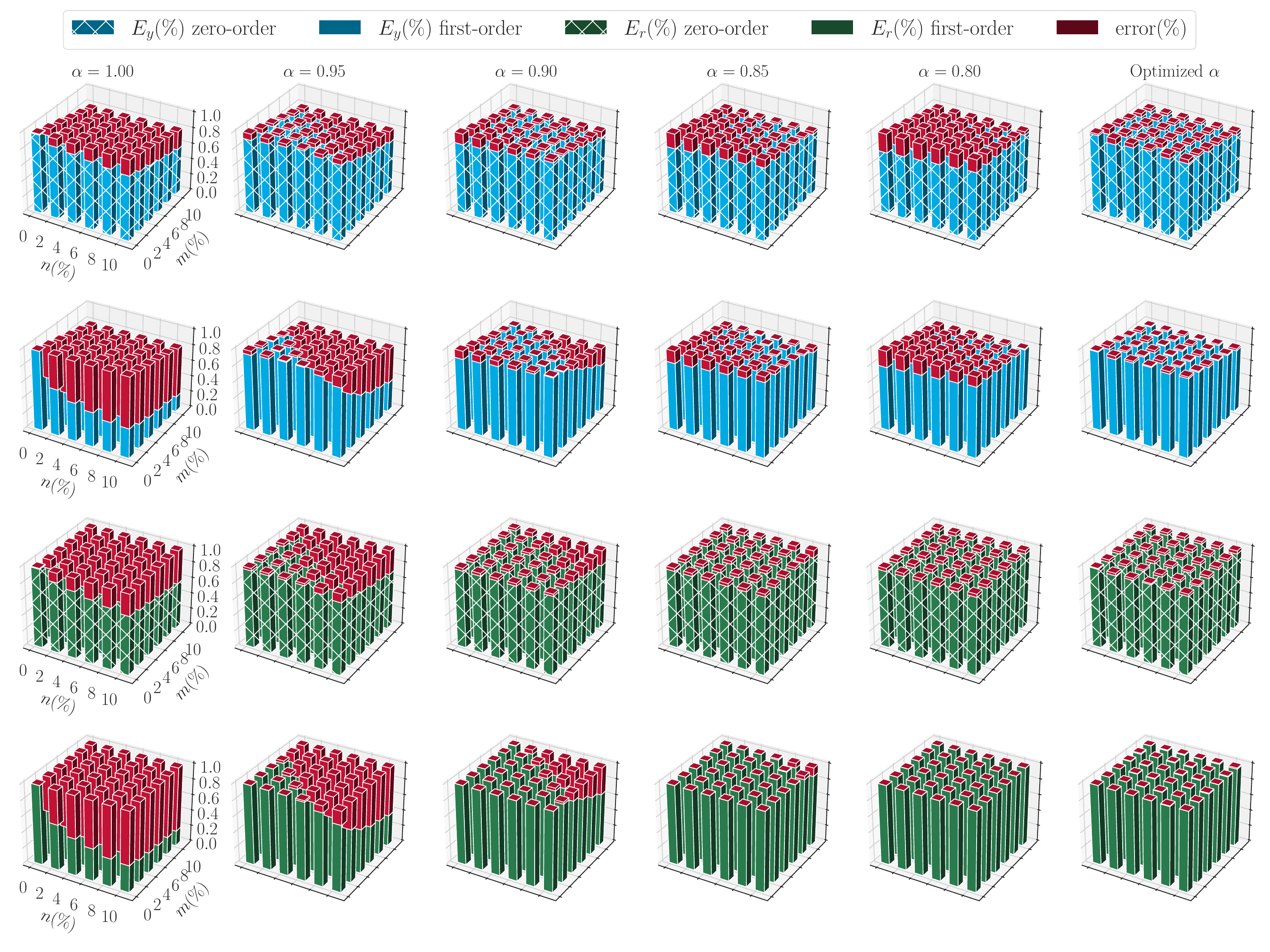}
  \caption{Case 2: mean yield and reliability rates and errors for different relative transport numbers $\alpha$,  and replacement percentages $n\%$ and $m\%$. Frame time step: 0.001 seconds.}
  \label{Burgers_replacement_performance}
\end{figure}

Fig. \ref{Burgers_replacement_performance} summarizes the particle matching results using the particle position collected every 0.001 seconds.
In each sub-figure, we plot yield or reliability rates against $n\%$ and $m\%$.
The first two rows plot the yield rates for zero and first-order prediction algorithms.
The last two rows plot the reliability rates for the two algorithms.
Each column summarizes the results using the same relative transport number $\alpha$, and its value decreases moving rightwards.
The last column plots the results of optimized $\alpha^{\ast}$.

To examine the matching performance of $\alpha$, we should check if it achieves the highest yield rate while maintaining a reliability rate of 100\%. 
We should mentioned that, due to the random replacement, the yield rate $E_r$ is approximately upper bounded by $(100-2n) \%$, i.e., 100, 96, 92, 88, 84, and 80\% for $n=0, 2, 4, 6, 8, 10$, respectively.
The reliability rates are always upper bounded by 100\%.
We plot the gaps of the yield and reliability rates to their upper bounds in red.

The plot shows that a larger $\alpha$, i.e., matching more particles, is not necessarily better when facing imperfect reconstructed particle position estimates.
For example, in the first column, when $\alpha=1$, the reliability rates are small.
In this case, we are overconfident about the position reconstruction results and intend to match misidentified particle positions between frames.
Smaller $\alpha$, e.g.,
0.85 (the fourth column) or 0.8 (the fifth column), sacrifices yield rates but is more robust.
Another noticeable result is that the first-order prediction method performs worse than the zero-order prediction algorithm for large $\alpha$.
For example, in the first column, the yield and reliability rates of the first-order method are lower than those of the zero-order method.
A potential explanation is that the accuracy of the first-order prediction relies on the correctness of the matched pairing in the last frame.
When reconstructed particle positions are not perfect, 
the errors propagate through the first-order prediction steps, which deteriorate the matching performance.
Therefore, optimizing the hyperparameter to identify the best relative transport number is essential in this case. 
In the last column, our optimized $\alpha^{\ast}$ achieves significant improvements since both yield and reliability rates are close to their upper bounds. 

\paragraph{Case 3: Random replacement plus jittering}
Last, we study the matching performance for the most challenging and practical case where the remaining reconstructed particle position estimates are also jittered after random replacement.
The jittering level is a fraction of the mean particle displacement $\bar{d} = \bar{u}\Delta t$ between two frames.
We experimented with six jittering levels $\delta \times \bar{d}$ for $\delta=0, 0.2, 0.4, 0.6, 0.8$, and $1.0$,
and perturbed the means of Gaussian by adding random samples following the uniform distribution defined in the cube $\left[-\delta \times \bar{d},\delta \times \bar{d}  \right]^3$.
The replacement levels $(n, m)$ are the same as case two, and we only report the result for the optimized relative transport number $\alpha^{\ast}$.

Fig. \ref{Burgers_replacement_noise_perform}
showcases the matching performance and error.
Each sub-figure plots either yield or reliability rates against $\delta$, $n\%$ and $m\%$ for optimized $\alpha^{\ast}$.
The first two rows plot the yield rates for zero and first-order prediction algorithms.
The last two rows plot the reliability rates for the two algorithms.
Each column summarizes the results using the same jittering level $\delta$, and its value increases moving rightwards.
In addition to the same conclusion in case 2,
this time, we observe for small jittering levels, such as $\delta=0.2$ and $0.4$, the first-order method still outperforms the zero-order method.
However, this is not the case when the jittering level becomes larger.
A plausible explanation is that the first-order method makes less accurate position predictions due to the larger random jitters. 
Therefore, choosing the order of position prediction depends on the accuracy of the particle position reconstruction.
\begin{figure}[!htb]
    \includegraphics[scale=0.32]{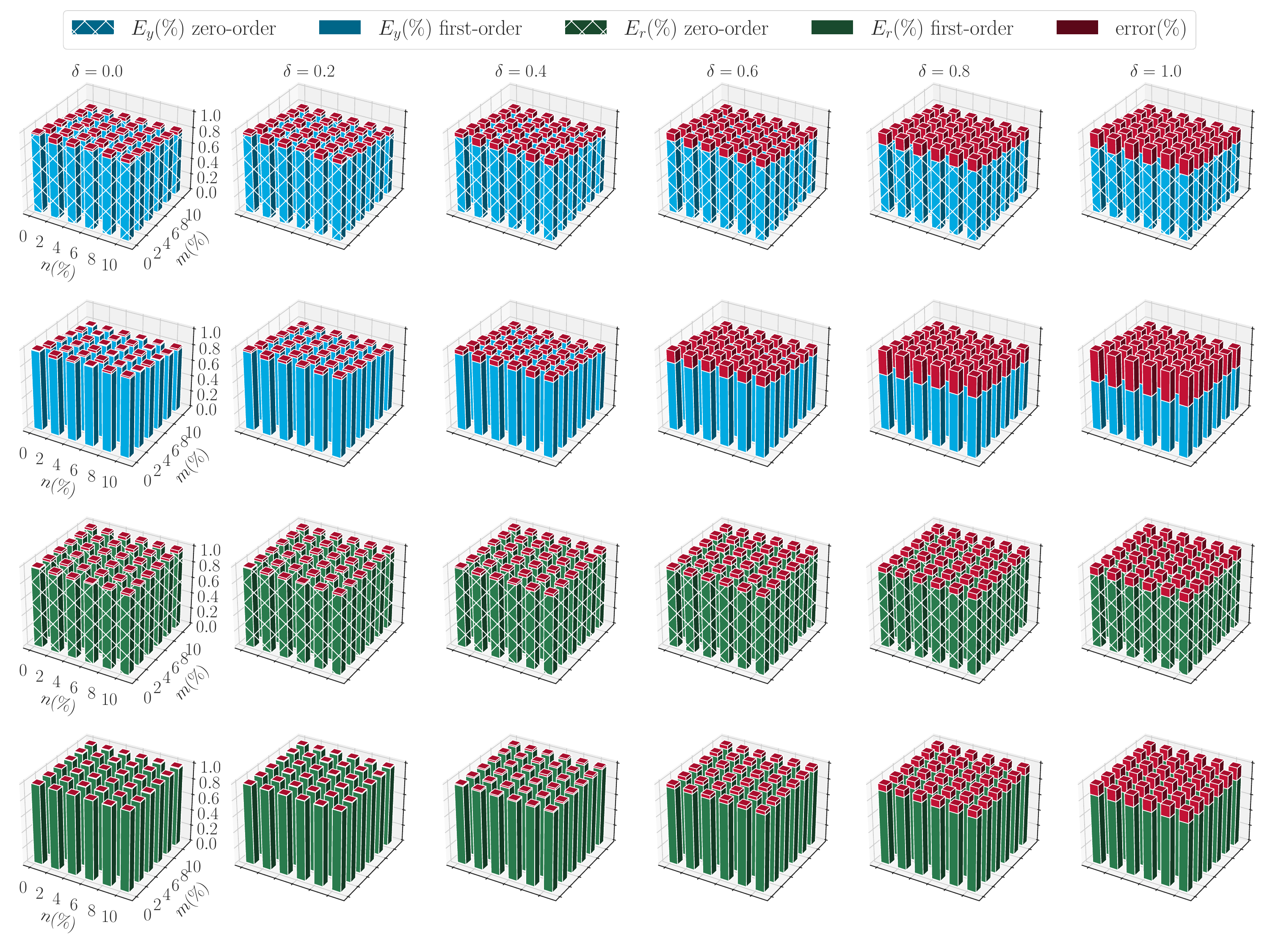}
  \caption{Case 3:  mean yield and reliability rates and errors using optimized relative transport number $\alpha^{\ast}$  for different jittering levels $\delta$, and replacement percentages $n\%$ and $m\%$. Frame time step: 0.001 seconds.}
  \label{Burgers_replacement_noise_perform}
\end{figure}

\subsection{Experimental example: cerebral aneurysm}

We demonstrate our tracking algorithm using the experimental cerebral aneurysm dataset. \citet{brindise2019multi} performed in-vitro experiments, where particle images were captured using four high-speed cameras with one center camera of $0^\circ$ angle from the geometry plane and the other three of about $30^{\circ}$ angles.
1216 $\times$ 1224 pixels time-resolved images were recorded at 2000 Hz.
Lagrangian tracks were identified using the shake-the-box \citep{schanz2016shake} algorithm.

For our purpose,
we first use the BVR algorithm~\citep{hansstochastic} to reconstruct the posterior distributions of particle positions at each frame using recorded images from the four cameras.
Then, we identify stochastic Lagrangian tracks using our first-order prediction PWD formulation of unbalanced optimal transport.
We showcase the reconstruction results in three plots.
In Fig. \ref{mean_compare_cerebral}, 
we plot the reconstructed particle position means using BVR in the left column, and the reconstructed track position and velocity means using PWD in the right column.
Notice that the velocity means are plotted at the location of the position means.
The identified tracks are visually reasonable compared to the reconstruted particle positions.
In Fig. \ref{standard_cerebral_position}, we plot the total standard deviations, i.e., the squared root of the sum of the variances of $x$, $y$, and $z$ coordinates.
All standard deviations are plotted at the location of the position means.
In the left column, we plot the result for the interior particles whose distances to the interior wall surface are greater than 1mm.
The right column shows the result for the particles near the interior wall surface.
We use pink color to emphasize high uncertainty.
As the color in the right column is slightly higher than the left column color, the uncertainty of the reconstructed particle positions near the wall exhibits a slightly higher uncertainty.
Usually, the reconstruction result shows more error and uncertainty in the regions that are closer to the wall, and our reconstruction result agrees with this fact.
Fig. \ref{standard_cerebral_velocity} plots the uncertainty of the velocities of the reconstructed Lagrangian tracks.
As before, the left and right columns plot the result for the tracks that are greater and within 1mm distance to the interior wall surface, respectively.
The velocity standard deviations are plotted at the location of the track means.
We can observe a similar result that the velocity uncertainty near the wall is slightly higher than the uncertainty in the interior region.

\begin{figure}[!htpb]
  \centering
\includegraphics[scale=0.19]{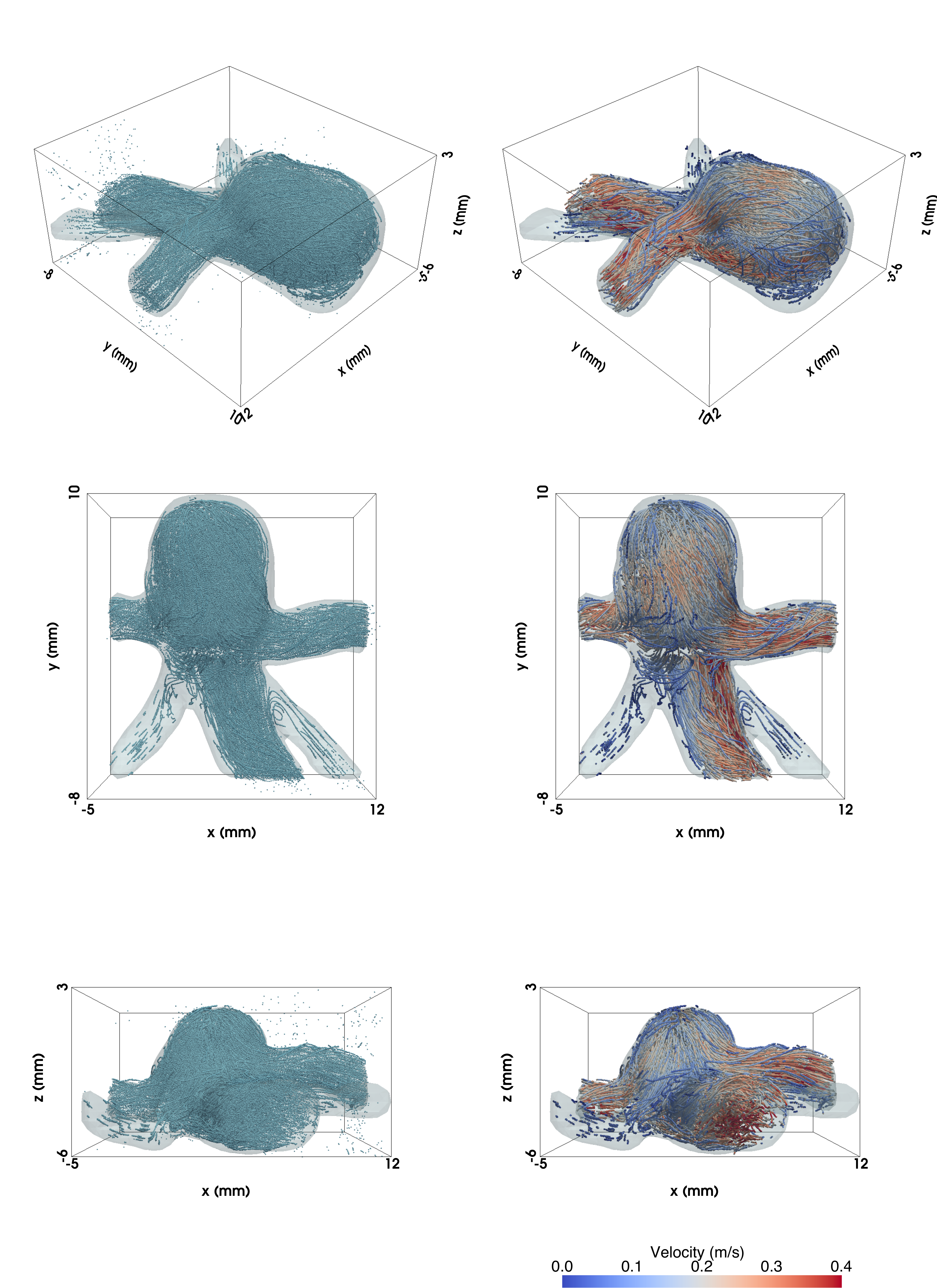}
  \caption{Cerebral aneurysm flow experimental data \citep{brindise2019multi}. Left column: reconstructed particle position means across all frames using BVR. Right column:
  reconstructed track position and velocity means using PWD.}
  \label{mean_compare_cerebral}
\end{figure}

\begin{figure}[!htpb]
  \centering
\includegraphics[scale=0.19]{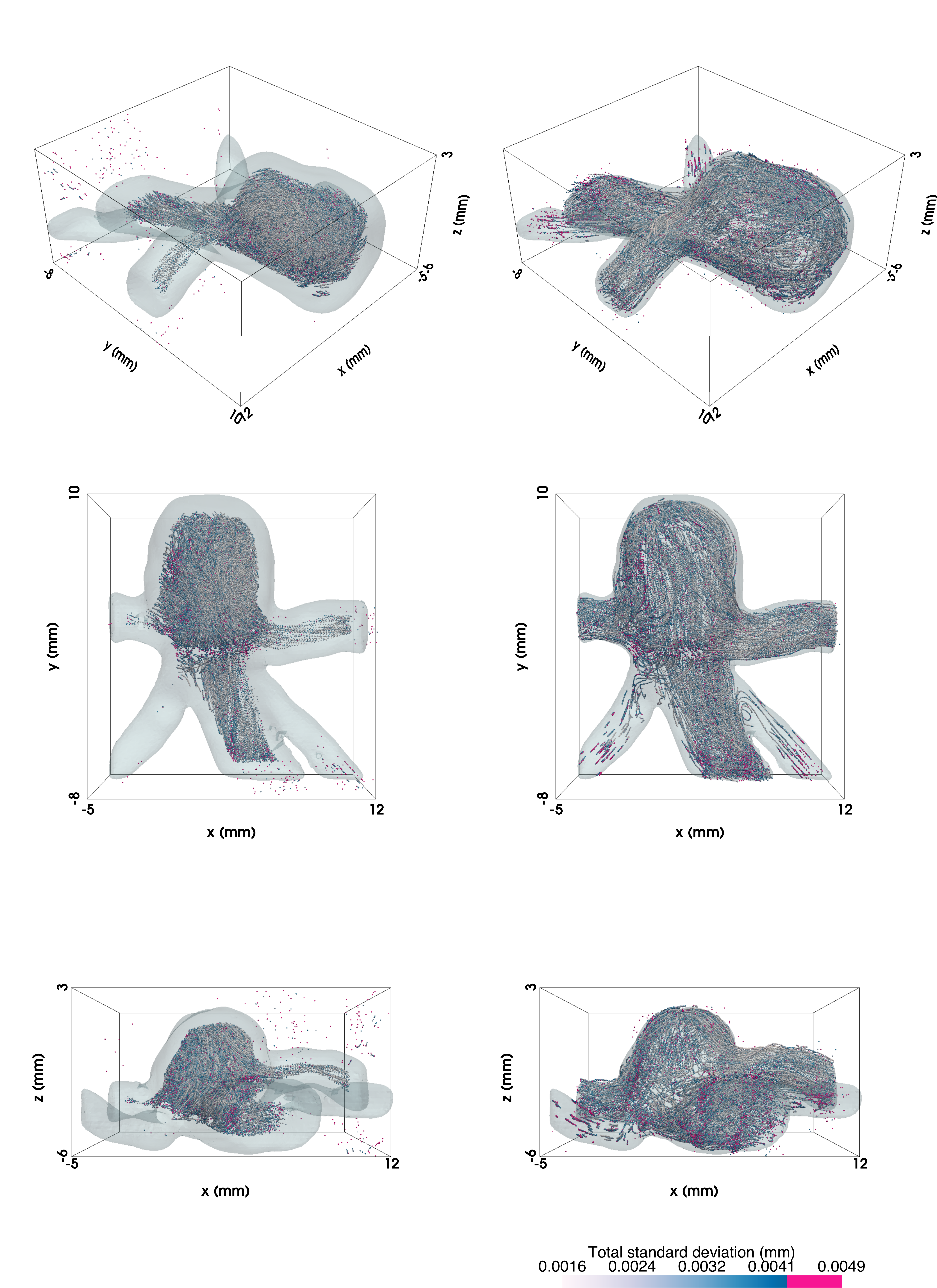}
  \caption{
  Cerebral aneurysm flow experimental data \citep{brindise2019multi}. Left column: total standard deviations of the reconstructed particle positions across all frames using BVR in the interior region (particle distances to the wall are greater than 1 mm). Right column: 
  total standard deviations near the wall (less than 1mm).
  }
  \label{standard_cerebral_position}
\end{figure}

\begin{figure}[!htpb]
  \centering
\includegraphics[scale=0.19]{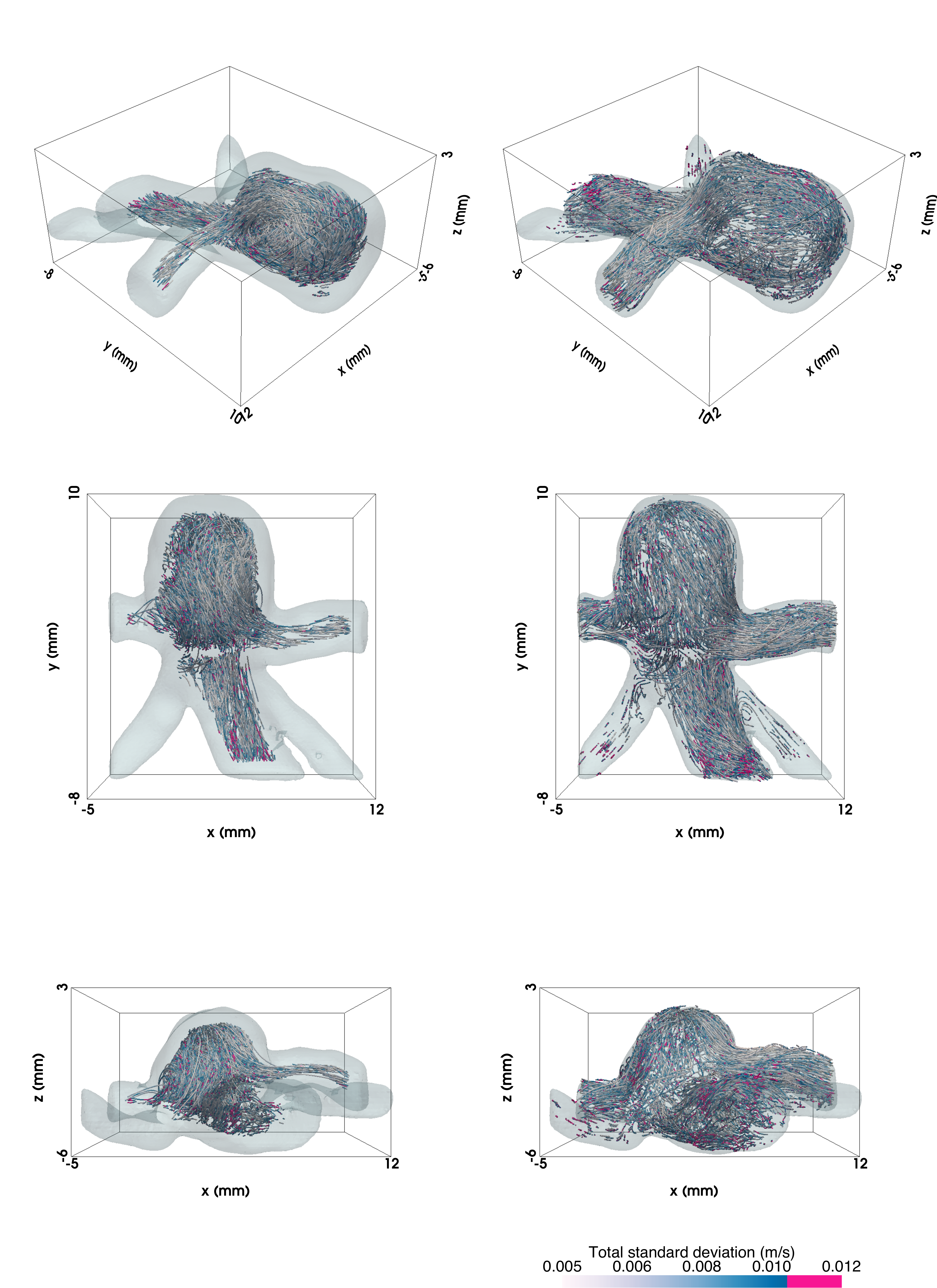}
  \caption{
  Cerebral aneurysm flow experimental data \citep{brindise2019multi}. Left column: total standard deviations of the reconstructed track velocities using PWD in the interior region (track distances to the wall are greater than 1 mm). Right column: 
  total standard deviations near the wall (less than 1mm).
  }
  \label{standard_cerebral_velocity}
\end{figure}

\section{Conclusions}
\label{section:conclusions}

In this paper, we develop a stochastic particle tracking method based on unbalanced optimal transport theory.
The method optimizes a transport plan to move the reconstructed particles (represented by Gaussian distributions) between two frames in the most efficient manner.
Since the number of reconstructed particles between two frames are rarely same, we employ an unbalanced optimal transport formulation.
We first formulate an integer unbalanced optimal transport problem to enforce one-to-one pairing of the Gaussian estimation of the particle positions.
We then relax this integer linear programming problem using the partial Wasserstein distance.
We prove the minimizer of the relaxed formulation is identical to the non-relaxed formulation.
This allows us to solve the original problem in polynomial time.
The unbalanced optimal transport formulation has a hyper-parameter, i.e., the total number of particles that need to be matched between two frames.
We develop a hyper-parameter optimization method to determine the most likely value for it.
This method first solves the unbalanced optimal transport with different hyper-parameters in parallel.
Then we find the largest hyper-parameter that keeps the estimated reliability rate close to 1.
We verify the developed stochastic particle tracking method using a synthetic 3D Burgers vortex example.
We evaluate our method's tracking performance in three cases of different difficulties.
The hyper-parameter optimization strategy significantly improves the yield and reliability rates of zero and first-order prediction methods.
Usually, the first-order prediction method achieves superior performance than the zero-order prediction method.
However, this is not true when the jittering level is high. 
Finally, we demonstrate and validate the developed stochastic particle tracking method using a cerebral aneurysm flow experimental example.
In addition to showcase the reconstructed position and velocity means,  we highlight the ability of our method to quantify the reconstruction uncertainty.
In particular, we plot the velocity uncertainties in the interior region and near the interior wall surface region, respectively.
The reconstruction result shows a slightly greater velocity uncertainty of the particles near the wall.

In summary, the developed method shows a convincing and promising capability to improve particle tracking accuracy and robustness.
The method also adds an additional advantage of quantifying reconstructed position and velocity uncertainties over the conventional method.
\appendix

\section{Proof of Theorem~\ref{theorem:1}}
\label{appendix:proof}
The essential idea of the proof is to show that PWD and IUOT constraints have the same extreme points.
Then, by the fundamental theorem of linear programming \citep{bertsekas2009convex},
linear programming attains a minimum at some extreme point and we conclude the theorem.

In the proof, we think $\gamma\in \mathbb{R}^{NM}$ as the vectorized transport plan matrix.
First we reformulate PWD constraints in Eqs.~(\ref{uot}) into the standard linear programming constraint format 
$
\left\{
\gamma
\middle\vert
A\gamma\le b, 
\gamma\ge 0
\right\}
$.
For an $N$ by $M$ dimensional transport plan matrix, we have:
\newenvironment{lbmatrix}[1]
  {\left[\array{@{}*{#1}{c}@{}}}
  {\endarray\right]}
\setcounter{MaxMatrixCols}{100}
\begin{align*}
    &A=
    \\
    &
    \begin{lbmatrix}{1}
        \\
        \\
        \\
        A_r\\
        \\
        \\
        \\
        \hdashline
        \\
        \\
        A_c\\
        \\
        \\
        \hdashline
        A_1
        \\
        \hdashline
        A_{-1}
    \end{lbmatrix}
    = 
    \begin{bmatrix}
    1&1&\cdots&1 & &&& && &&& && &&&
    \\
     &&& & 1&1&\cdots&1 && &&& && &&&
    \\
    &&& & &&& &\ddots&  &&& && &&&
    \\
    &&& & &&& && 1&1&\cdots&1 && &&&
    \\
    &&& & &&& && &&& &\ddots& &&&
    \\
    &&& & &&& && &&& && 1&1&\cdots&1
    \\
    1 &&&& 1 &&&&& 1&&&&& 1 &&&
    \\
    & 1 &&&& 1 &&&&& 1 &&&&& 1 &&
    \\
    && \ddots &&&& \ddots &&\cdots &&& \ddots &&\cdots &&&\ddots &
    \\
    &&&1 &&&& 1 &&&&& 1 &&&&&1
    \\
    1&1&\cdots&1 & 1&1&\cdots&1 &\cdots& 1&1&\cdots&1 &\cdots & 1&1&\cdots&1
    \\
    -1&-1&\cdots&-1 &-1&-1&\cdots&-1 &\cdots& -1&-1&\cdots&-1 &\cdots &-1&-1&\cdots&-1
    \end{bmatrix}
    .
\end{align*}

The block matrix $A_r$ and $A_c$ correspond to the row and column mass constraints, respectively.
We call rows in $A_r$ source rows and rows in $A_c$ target rows.
We rewrite the equality constraint $\sum_{i=1}^{NM}\gamma_i = N_p$ into the two inequality constraints $\sum_{i=1}^{NM}\gamma_i  \le N_p$ and $-\sum_{i=1}^{NM}\gamma_i  \le -N_p$, which are represented by
rows $A_1$ and $A_{-1}$.
Then by the Hoffman and Kruskal theorem \citep{hoffman2010integral,conforti2014integer}, the polyhedron 
$
\left\{
\gamma \middle\vert
A\gamma \le b, \gamma\ge 0
\right\}
$
is integral, i.e., all vertices are integer, for every $b\in \mathbb{Z}^{NM}$ if and only if $A$ is totally unimodular.
A matrix is totally unimodular (TU) \citep{hoffman2010integral} if the determinant of its every square submatrix is 0, 1 or -1.
In our case, $b$ is an integer vector since $N_p$ is integer.
If $A$ is TU, then by the fundamental theorem of linear programming \citep{bertsekas2009convex}, the PWD problem guarantees to solve the binary linear programming problem. 
\begin{lemma}
    A is totally unimodular.
    \label{lemma1}
\end{lemma}
\begin{proof}[Proof of Lemma \ref{lemma1}]
    First, notice that we just need to prove that $A^{\prime}=\left[
    A_r^T, A_c^T, A_1^T
    \right]^T$ is TU by dropping the last row $A_{-1}$ since multiplying a row by -1 and removing a repeated row preserve TU \cite{Kuniga}.
    Our goal is to prove that every square submatrix of $A^{\prime}$ has a determinant of 0, 1, or -1.
    It is already known that the transport coefficient matrix $\left[
    A_r^T, A_c^T
    \right]^T$ is TU \citep{lpch}.
    So we just need to consider the square submatrix that contains the special row $A_1$.
    We follow the proof steps of Theorem 4 about the network flow problem in Professor Mikhail Lavrov's lecture notes \citep{Mikhail}.

    Suppose the selected square submatrix contains totally $k+m+1$ rows and columns, which are the $k$ source rows $s_1, \dots, s_k$ chosen from $A_r$,  $m$ column rows $t_1, \dots, t_m$ chosen from $A_c$, and the special row $A_r$.
    The columns
    $
    c_{1:k+m+1}
    $
    can be arbitrary.
    We denote this square submatirx 
    \begin{align*}
        A^{\prime}
        \left(
        s_{1:k}, t_{1:m}, A_1;
        c_{1:k+m+1}
        \right)
        =
        \begin{bmatrix}
            A_r\left(
            s_{1:k}; c_{1:k+m+1}
            \right)
            \\
            A_c\left(
            t_{1:m}; c_{1:k+m+1}
            \right)
            \\
            \mathbf{1}_{1\times (k+m+1)}
        \end{bmatrix}_{(k+m+1)\times (k+m+1)}.
    \end{align*}
    First, we do standard reduction to consider each row at least contains two 1s.
    If not, we can do cofactor expansion along this row to reduce it to a smaller square submatrix with the same determinant up to the sign.

    From the structure of the source rows $s_1, \cdots, s_k$, we observe every column in 
    $
    A_r\left(
    s_{1:k}; c_{1:k+m+1}
    \right)
    $ 
    can not contain more than one 1.
    Moreover, from the assumption every row contains at least two 1s, 
    so the submatrix 
    $
        A^{\prime}
        \left(
        s_{1:k}, t_{1:m}, A_1;
        c_{1:k+m+1}
        \right)
    $
    must contains at least $2k$ columns, namely $k+m+1\ge 2k$.
    We can use the same argument for the column rows $t_1, \cdots, t_m$ to deduce the inequality
    $
    k+m+1\ge 2m
    $.
    From these two inequalities, we have $m-1 \le k \le m+1$, so there are only three possible cases: $k=m$, $k=m-1$, and $k=m+1$.

    If $k=m+1$, the submatrix is $2k$ by $2k$.
    This means that the special row 
    $
    \mathbf{1}_{1\times (k+m+1)}
    $
    has $2k$ 1s.
    Since, from the assumption, each row  contains at least two 1s and the submatrix $A_r$ is $k$ by $2k$, we have each row of $A_r$ has exactly two 1s.
    Then it is easy to see the special row 
    $
    \mathbf{1}_{1\times (k+m+1)}
    $
    is the sum of the source rows in 
    $
    A_r\left(
    s_{1:k}; c_{1:k+m+1}
    \right)
    $
    .
    This means the special row is linearly dependent on the source rows.
    Hence, the determinant of
    $
    A^{\prime}
        \left(
        s_{1:k}, t_{1:m}, A_1;
        c_{1:k+m+1}
        \right)
    $
    is 0.
    If $k=m-1$, the submatrix is $2m$ by $2m$.
    We can use the same argument using the target rows 
    $
    A_c\left(
    t_{1:m}; c_{1:k+m+1}
    \right)
    $
    to see the the determinant is also 0.

    So we just need to prove the case of $k=m$ and the submatrix is $2k+1$ by $2k+1$.
    There are two cases.
    
    First, we consider the case of $k-1$ source rows
    in 
    $
    A_r
    \left(
    s_{1:k}; c_{1:k+m+1}
    \right)
    $
    containing two 1s and the remaining source row containing three 1s.
    If this is the case, the special row 
    $\mathbf{1}_{1\times(k+m+1)}$
    is linearly dependent on all source rows, and the determinant is 0.
    Similarly, if one of the target rows contains three 1s, the determinant is 0.

    The second case is all source and target rows contain two 1s.
    In this case, only two things can happen.
    The first is there is a zero column except for the special row. 
    Since column permutation preserves TU \citep{Kuniga},
    without loss of generality, we assume the last column is such a special column:
    \begin{align*}
        A^{\prime}
        \left(
        s_{1:k}, t_{1:m}, A_1;
        c_{1:k+m+1}
        \right)
        =
        \begin{bmatrix}
            A_r\left(
            s_{1:k}; c_{1:k+m}
            \right)
            &
            \mathbf{0}_{k\times 1}
            \\
            A_c\left(
            t_{1:m}; c_{1:k+m}
            \right)
            &
            \mathbf{0}_{m\times 1}
            \\
            \mathbf{1}_{1\times (k+m)}
            &1
            \end{bmatrix}_{(k+m+1)\times (k+m+1)}.
    \end{align*}
    In this case, we do the cofactor expansion along this special column using the 1 corresponding to the special row.
    Then we have
    \begin{align*}
        \operatorname{det}
        \left(
        A^{\prime}
        \right)
        =
        (-1)^{2(k+m+1)}
        \operatorname{det}
        \left(
        \begin{bmatrix}
        A_r\left(
        s_{1:k}; c_{1:k+m}
        \right)
        \\
        A_c\left(
        t_{1:m}; c_{1:k+m}
        \right)
        \end{bmatrix}
        \right).
    \end{align*}
    Notice that
    \begin{align*}
        \operatorname{det}
        \left(
        \begin{bmatrix}
        A_r\left(
        s_{1:k}; c_{1:k+m}
        \right)
        \\
        A_c\left(
        t_{1:m}; c_{1:k+m}
        \right)
        \end{bmatrix}
        \right)
        = 0, -1, \ \text{or} \ 1
    \end{align*} 
     since it is a submatrix of the transport coefficient matrix $
    \left[
    A_r^T, A_c^T
    \right]^T
    $, which is already known to be TU.
    The second possibility is that there is a column containing one 1 except for the special row.
    Without loss of generality,
    let's assume this special column is due to the source rows.
    We can also assume the special column is the last column and the row contains this special 1 is the first row:
    \begin{align*}
        A^{\prime}
        \left(
        s_{1:k}, t_{1:m}, A_1;
        c_{1:k+m+1}
        \right)
        =
        \begin{bmatrix}
            0,\cdots, 1, \cdots 0, & 1
            \\
            A_r\left(
            s_{2:k}; c_{1:k+m}
            \right)
            &
            \mathbf{0}_{k-1\times 1}
            \\
            A_c\left(
            t_{1:m}; c_{1:k+m}
            \right)
            &
            \mathbf{0}_{m\times 1}
            \\
            \mathbf{1}_{1\times (k+m)}
            &1
            \end{bmatrix}_{(k+m+1)\times (k+m+1)}.
    \end{align*}
    In this case, we do the cofactor expansion along this special column:
    \begin{align*}
        \operatorname{det}
        \left(
        A^{\prime}
        \right)
        =
        (-1)^{k+m+2}
        \operatorname{det}\left(
        \begin{bmatrix}
             A_r\left(
            s_{2:k}; c_{1:k+m}
            \right)
            \\
            A_c\left(
            t_{1:m}; c_{1:k+m}
            \right)
            \\
            \mathbf{1}_{1\times (k+m)}
        \end{bmatrix}
        \right)
        +
        (-1)^{2(k+m+1)}
        \operatorname{det}
        \left(
        \begin{bmatrix}
             0,\cdots, 1, \cdots 0
             \\
              A_r\left(
            s_{2:k}; c_{1:k+m}
            \right)
            \\
             A_c\left(
            t_{1:m}; c_{1:k+m}
            \right)
        \end{bmatrix}
        \right).
    \end{align*}
    Using the same argument as before, the first determinant is 0 since the all 1s row is linearly dependent on the target rows in $A_c(t_{1:m}; c_{1:k+m})$.
    For the second determinant, we have
    \begin{align*}
        \operatorname{det}
        \left(
        \begin{bmatrix}
             0,\cdots, 1, \cdots 0
             \\
              A_r\left(
            s_{2:k}; c_{1:k+m}
            \right)
            \\
             A_c\left(
            t_{1:m}; c_{1:k+m}
            \right)
        \end{bmatrix}
        \right)=
        0, -1, \ \text{or}\ 1
    \end{align*}
    since it is a submatrix of the transport coefficient matrix $
    \left[
    A_r^T, A_c^T
    \right]^T
    $.

    Together, we have proved $A$ is TU.
\end{proof}

Now, the proof of Theorem \ref{theorem:1} is straightforward.
\begin{proof}[Proof of Theorem \ref{theorem:1}] 
By the Hoffman and Kruskal theorem \citep{hoffman2010integral,conforti2014integer}, the polyhedron 
    $
    \left\{
    \gamma \middle\vert
    A\gamma \le b, \gamma\ge 0
    \right\}
    $
    is integral, i.e., all extreme points (vertices) are integers, for every $b\in \mathbb{Z}^{NM}$ if and only if $A$ is TU.
     $A$ is TU from Lemma \ref{lemma1}, so the solution of PWD is integral due to the fundamental theorem of linear programming.

    We notice that the PWD constraints
    $\sum_{j=1}^M\gamma_{ij}\leq 1$,
    $\sum_{i=1}^N\gamma_{ij}\leq 1$,
     and
     $\gamma_{ij} \ge 0$ imply
     $0 \le \gamma_{ij}\le 1$.
     This implies that the solution of PWD must satisfy $\gamma_{ij}=\left\{0, 1\right\}$.
     Since the objective function and the first three constraints of IUOT and PWD are the same, we conclude $\gamma^{\ast}_{IUOT}=\gamma^{\ast}_{PWD}$.
\end{proof}

\section{Stochastic Reconstruction: First-order Prediction Full-rank Covariance Matrix Case}
\label{appendix: full cov}

In this section, we discuss the first-order prediction for the full-rank covariance matrix case.

For $i\in A_{\text{matched}}$, 
linear extension has the form
\begin{align*}
    R^{k+1|k}_i = R^k_i + 
    \left(
    R^k_i - R^{k-1}_{l(i)}
    \right), 
    \
    \text{where} \
    R^k_i\sim \mathcal{N}\left(m^{k}_{i}, \Sigma^{k}_{i}\right) \ \text{and} \
    R^{k-1}_{l(i)} \sim 
    \mathcal{N}\left(m^{k-1}_{l(i)}, \Sigma^{k-1}_{l(i)}\right).
\end{align*}
We  assume
$
\mathcal{N}\left(m^{k}_{i}, \Sigma^{k-1}_{i}\right)
$ 
and 
$
\mathcal{N}\left(m^{k-1}_{l(i)}, \Sigma^{k-1}_{l(i)}\right)
$ 
are independent.
So 
$
R^{k+1|k}_i
$
follows
$
\mathcal{N}\left(m^{k+1|k}_{i}, \Sigma^{k+1|k}_{i}\right)
$, and we denote $\mathbb{C}$ the covariance operator to derive
\begin{align*}
    m^{k+1|k}_{i} &= 
    \mathbb{E}
    \left[
    2R^{k}_i- R^{k-1}_{l(i)}
    \right]
    =
    2m^k_i - m^{k-1}_{l(i)}, \
    \text{and}\
    \Sigma^{k+1|k}_i
    =\mathbb{C}
    \left[
    2R^{k}_i- R^{k-1}_{l(i)}
    \right]
    =
     4\Sigma^k_{i} + \Sigma^{k-1}_{l(i)}.
\end{align*}

For $i\in A_{\text{unmatched}, e}$, we have to estimate the translation, rotation and stretch of Gaussian ellipse using $B^{\ast}_{\epsilon}(\pi)$-particles.
For the translation, we use Wasserstein distance weighted predicted displacements of $B^{\ast}_{\epsilon}(\pi^k_i)$-particles to estimate the predicted mean position:
\begin{align*}
    m^{k+1|k}_i
    =
    m^k_i + 
     \dfrac{
     \sum_{j\in A_{\text{matched}, ei}}
    W_{2}\left(
    \mathcal{N}(m^k_j, \Sigma^k_j),
    \
    \mathcal{N}(m^{k}_i, \Sigma^{k}_i)
    \right)
    \left(
    m^k_j -m^{k-1}_{l(j)}
    \right)
    }{
     \sum_{j\in A_{\text{matched}, ei}}
   W_{2}\left(
    \mathcal{N}(m^k_j, \Sigma^k_j),
    \
    \mathcal{N}(m^{k}_i, \Sigma^{k}_i)
    \right)
    }.
\end{align*}
The covariance matrix accounts for the rotation and stretch.
First, we study the covariance transformation of a single $B^{\ast}_{\epsilon}(\pi)$-particle $\pi^k_j$.
The covariance matrix $\Sigma^k_j$ has the eigendecomposition 
$\Sigma^k_j = U_j^k\Lambda^k_j {U_j^k}^T$, where we organize the eigenvalues in the descending order.
Similarly, the covariance matrix $\Sigma^{k+1}_{l(j)}$ of the matched pair of $\pi^k_j$ has the eigendecomposition
$\Sigma^{k+1}_{l(j)} = U_{l(j)}^{k+1}\Lambda^{k+1}_{l(j)} {U_{l(j)}^{k+1}}^T$
The  transformation from $\Sigma^k_j$ to $\Sigma^{k+1|k}_j$ can be summarized by the rotation 
$
Q\left(
\alpha_j^{k+1|k},
\beta_j^{k+1|k},
\gamma_j^{k+1|k}
\right)
$ 
and diagonal scaling $D_j^{k+1|k}$ matrices:
\begin{align*}
\Sigma^{k+1|k}_j = Q\left(
\alpha_j^{k+1|k},
\beta_j^{k+1|k},
\gamma_j^{k+1|k}
\right) U_j^k D_j^{k+1|k} \Lambda^k_j {U_j^k}^T Q^T\left(
\alpha_j^{k+1|k},
\beta_j^{k+1|k},
\gamma_j^{k+1|k}
\right).
\end{align*}
The rotation matrix can be computed by $Q = U^{k+1|k}_{l(j)}{U^{k}_j}^T$.
We specify the rotation is performed in the sequence of about axes $x$, $y$, and $z$ with the Euler angles $\alpha_j^{k+1|k},
\beta_j^{k+1|k}$, and
$\gamma_j^{k+1|k}$.
So the rotation matrix can be decomposed into the product of three Givens rotation \citep{golub2013matrix}:
\begin{align*}
    &Q\left(
    \alpha_j^{k+1|k},
    \beta_j^{k+1|k},
    \gamma_j^{k+1|k}
    \right)
    \\
    =&
    \begin{bmatrix}
    \cos{\gamma_j^{k+1|k}} & -\sin{\gamma_j^{k+1|k}} & 0
    \\
    \sin{\gamma_j^{k+1|k}} & \cos{\gamma_j^{k+1|k}} & 0
    \\
    0 & 0 & 1
    \end{bmatrix}
    \begin{bmatrix}
        \cos{\beta_j^{k+1|k}} & 0 & \sin{\beta_j^{k+1|k}}
        \\
        0 & 1 & 0 
        \\
        -\sin{\beta_j^{k+1|k}} & 0 & \cos{\beta_j^{k+1|k}}
    \end{bmatrix}
    \begin{bmatrix}
        1 & 0 & 0
        \\
        0 & \cos{\alpha_j^{k+1|k}} & -\sin{\alpha_j^{k+1|k}}
        \\
        0 & \sin{\alpha_j^{k+1|k}} & \cos{\alpha_j^{k+1|k}}
    \end{bmatrix}.
\end{align*}
The diagonal scaling matrix is computed by dividing the diagonal elements in $\Lambda^{k+1}_{l(j)}$ by the diagonal elements in $\Lambda^{k}_{j}$.

Then we use Wasserstein distance weighted rotation angles and scaling factors of $B^{\ast}_{\epsilon}(\pi)$-particles to make the covariance prediction for particle $\pi^k_i$:
\begin{align*}
    \theta^{k+1|k}_{i, B^{\ast}_{\epsilon}(\pi)}
    &=
    \dfrac{
             \sum_{j\in A_{\text{matched}, ei}}
            W_{2}\left(
            \mathcal{N}(m^k_j, \Sigma^k_j),
            \
            \mathcal{N}(m^{k}_i, \Sigma^{k}_i)
            \right)
            \theta^{k+1|k}_j
            }{
             \sum_{j\in A_{\text{matched}, ei}}
           W_{2}\left(
            \mathcal{N}(m^k_j, \Sigma^k_j),
            \
            \mathcal{N}(m^{k}_i, \Sigma^{k}_i)
            \right)
            },\quad \text{for}\ \theta = \alpha, \beta, \text{and}\ \gamma,
            \\
    D^{k+1|k}_{i, B^{\ast}_{\epsilon}(\pi)}
    &=
    \dfrac{
             \sum_{j\in A_{\text{matched}, ei}}
            W_{2}\left(
            \mathcal{N}(m^k_j, \Sigma^k_j),
            \
            \mathcal{N}(m^{k}_i, \Sigma^{k}_i)
            \right)
            D^{k+1|k}_j
            }{
             \sum_{j\in A_{\text{matched}, ei}}
           W_{2}\left(
            \mathcal{N}(m^k_j, \Sigma^k_j),
            \
            \mathcal{N}(m^{k}_i, \Sigma^{k}_i)
            \right)
            }.
\end{align*}
For $i\in A_{\text{unmatched}, n}$, we keep the zero-order prediction.

In summary, for the stochastic reconstruction, we have
\begin{align*}
    m^{k+1|k}_i
    =
    \begin{cases}
        m^{k}_i + \left(
        m^k_i - m^{k-1}_{l(i)}
        \right), &\forall i \in A_{\text{matched}},
        \\
        m^k_i + 
         \dfrac{
         \sum_{j\in A_{\text{matched}, ei}}
        W_{2}\left(
        \mathcal{N}(m^k_j, \Sigma^k_j),
        \
        \mathcal{N}(m^{k}_i, \Sigma^{k}_i)
        \right)
        \left(
        m^k_j -m^{k-1}_{l(j)}
        \right)
        }{
         \sum_{j\in A_{\text{matched}, ei}}
       W_{2}\left(
        \mathcal{N}(m^k_j, \Sigma^k_j),
        \
        \mathcal{N}(m^{k}_i, \Sigma^{k}_i)
        \right)
        }, & \forall i \in A_{\text{unmatched}, e},
        \\
        m^k_i, & \forall i \in A_{\text{unmatched}, n},
    \end{cases}
\end{align*}
and
\begin{align*}
    &\Sigma^{k+1|k}_i
    \\
    =&
    \begin{cases}
         4\Sigma^k_{i} + \Sigma^{k-1}_{l(i)},
         & \forall i \in A_{\text{matched}},
         \\
         Q\left(
        \alpha_{i, B^{\ast}_{\epsilon}(\pi)}^{k+1|k},
        \beta_{i, B^{\ast}_{\epsilon}(\pi)}^{k+1|k},
        \gamma_{i, B^{\ast}_{\epsilon}(\pi)}^{k+1|k} 
        \right) 
        U_i^k 
        D^{k+1|k}_{i, B^{\ast}_{\epsilon}(\pi)}
        \Lambda^k_i 
        {U_i^k}^T Q^T\left(
        \alpha_{i, B^{\ast}_{\epsilon}(\pi)}^{k+1|k},
        \beta_{i, B^{\ast}_{\epsilon}(\pi)}^{k+1|k},
        \gamma_{i, B^{\ast}_{\epsilon}(\pi)}^{k+1|k}
        \right), &\forall i \in A_{\text{unmatched}, e},
        \\
        \Sigma^k_i, & \forall i \in A_{\text{unmatched}, n}.
    \end{cases}
\end{align*}

We only cover the theoretical formula.
The numerical implementation is computationally expansive since it requires performing nested eigenvalue decomposition to compute the squared 2-Wasserstein distance in \qref{normal_was_full} and the rotation angles.

\bibliographystyle{unsrtnat}

\end{document}